\documentclass[conference]{IEEEtran}

\usepackage{cite}

\usepackage{soul}

\usepackage[T1]{fontenc}
\usepackage{graphicx}
\usepackage[english]{babel}
\usepackage{amsthm}
\usepackage{amsmath,amssymb}
\usepackage{ stmaryrd }
\usepackage{booktabs}
\usepackage{xfrac}
\usepackage{xcolor}
\usepackage{mathtools,amsfonts,mathpartir}
\usepackage{pifont}
\usepackage{array}
\usepackage{algorithm}
\usepackage[noend]{algpseudocode}
\usepackage{listings}
\usepackage[inline]{enumitem}
\usepackage{xspace}
\usepackage{semantic}
\usepackage{mathrsfs}
\usepackage{verbatim}
\usepackage[normalem]{ulem}
\usepackage{cleveref}
\usepackage{multirow} 
\usepackage{thm-restate}
\usepackage{prftree}

\include{IoTmacros}

\newcommand{\rulename}[1]{\mbox{\footnotesize\scshape(#1)}}
\usepackage[normalem]{ulem}

\usepackage{tikz}
\usetikzlibrary{positioning,fit,calc,backgrounds}
\usetikzlibrary{arrows.meta,calc,positioning,decorations.pathmorphing,fit,shapes,patterns,patterns.meta}
\usetikzlibrary{decorations.pathmorphing,positioning,calc}

 \newtheorem{lemma}{Lemma}
 
 \newtheorem{example}{Example}
 \newtheorem{definition}{Definition}
 \newtheorem{notation}{Notation}

\algnewcommand\AAnd{\textbf{and} }
\newcommand{\thetool}{\textsc{IoT:Poker}}

\lstdefinelanguage{term}{
    keywords = {\$, ?, -, iotpoker, reachOnly, sat, unsat}
}

\lstdefinestyle{term}{
	basicstyle=\footnotesize\ttfamily, 
	language=term,
	showstringspaces=false,
	tabsize=1,
	breaklines=true,
	breakatwhitespace=false,
}

\newcommand{\cmark}{\ding{51}}
\newcommand{\xmark}{\ding{55}}

\def\BibTeX{{\rm B\kern-.05em{\sc i\kern-.025em b}\kern-.08em
    T\kern-.1667em\lower.7ex\hbox{E}\kern-.125emX}}

\allowdisplaybreaks

\AddToHook{cmd/appendices/before}{%
	\crefalias{section}{appendix}%
	\crefalias{subsection}{appendix}
}

\newcommand{\permit}{\mathrel{\mid\mkern-3mu\sim}}

\begin{document}

\title{Checking Information Flow in Cloud-based IoT Access Control Policies\\
(Extended Version)
}

\author{
\IEEEauthorblockN{
Lorenzo Ceragioli,
Letterio Galletta,
Edoardo Lunati
}
\IEEEauthorblockA{
IMT School for Advanced Studies Lucca\\
Lucca, Italy\\
Emails: lorenzo.ceragioli@imtlucca.it, letterio.galletta@imtlucca.it, edoardo.lunati@imtlucca.it
}
}

\maketitle              %

\begin{abstract}
  Many cloud providers for IoT technologies offer 
  access control mechanisms whose proper configuration is 
  critical
  for security.
  However, verifying permissions in isolation is insufficient
  in a setting where devices have different levels of trust or are compartmentalised in various subsystems.
  This work analyses IoT access control policies to identify potential security vulnerabilities from unwanted information flow between devices. 
  To this end, we formally model  AWS IoT Core's components and define an information flow graph to capture the communication among devices permitted by the access control policies.
  We build a finite representation of the graph by leveraging an SMT solver,
  thus enabling the verification of information flow between devices.   
  We implement our approach in a tool called \thetool, 
  and assess it on a realistic scenario and several real-world policies.
\end{abstract}

\section{Introduction}\label{sec:introduction}

\IEEEPARstart{C}{loud} computing has emerged as a cornerstone for Internet of Things (IoT) deployments, offering tailored infrastructure and platform services: 
numerous cloud providers offer specialised Platform-as-a-Service and Infrastructure-as-a-Service solutions 
for IoT applications, e.g., AWS IoT Core~\cite{IoTCore} and Azure IoT Hub~\cite{IoTAzure}.  
These services allow 
developers to offload 
security and deployment responsibilities
onto cloud providers.

Regarding security, many providers supply developers with languages to define their access control policies (\emph{IoT policies}).
An IoT policy is a specification of what resources an IoT device can access, which actions it can perform, and under which conditions (e.g., publish an MQTT~\cite{MQTTSpec} message -- a type of action from
a popular IoT messaging protocol, see Section~\ref{sec:background}).
Properly configured policies are essential for the security of IoT systems.
Indeed, previous research has highlighted the susceptibility of IoT policies to misconfigurations and the significant risks they may cause~\cite{p-verifier}. 
However, preventing unauthorised access
is necessary but not sufficient to ensure security in a setting where devices have a different degree of confidentiality, trust 
and of robustness against attacks. 
Therefore, developers resort to compartmentalisation and information flow control to increase the security level of their systems. 
To protect integrity, they assign integrity levels to devices and prevent any direct or transitive dependency that would let lower-integrity devices influence higher-integrity ones ("no read down, no write up").
Similarly, to enforce confidentiality, they assign confidentiality levels to devices and disallow interactions that would allow devices to learn from higher levels ("no read up") or leak to lower levels ("no write down").
Ensuring both integrity and confidentiality requires isolating devices and removing unwanted dependencies.
We say that there is an information flow between devices $d_1$ and $d_n$ if some devices $d_2, d_3 \dots, d_{n-1}$ exist such 
that the device $d_i$ sends a message that is received by $d_{i+1}$~\cite{Gran,marc2,Guttman,Radhika,IFCIL}.

This paper addresses the problem of checking information flow in IoT access policies, namely, 
protecting critical devices from untrusted devices and thereby enforcing security properties, including confidentiality and integrity.
We provide a formal model that characterises how access to resources in an IoT-Cloud deployment is regulated via policies.
Then, we introduce a novel verification procedure to verify that a policy enjoys such information-flow properties and implement it in a tool called \thetool.
Formal verification techniques have been proven effective in detecting misconfigurations and verifying security properties in access control policies at various levels.
Indeed, prior work~\cite{p-verifier,Zelkova,BultanPerm,BultanQuacky,BultanRepair} studied the problem of permission misconfiguration in both cloud-based and IoT access policies, focusing on AWS Identity and Access Management (IAM) and AWS IoT Core. 
However, such works analyse policies in isolation to check the granted permissions and do not consider possible interactions among devices, so the possible information flow
(see \Cref{sec:ralated}). 
Here, we fill this gap by providing a triple contribution:
\begin{enumerate}[label=(\roman*)]
\item We introduce a formal model of IoT policies, specifically targeting AWS IAM policies and their evaluation within the AWS IoT Core infrastructure.
 	\item  Given a set of IoT policies, we analyse them and build an \emph{information flow graph} to capture the possible communication interaction between devices allowed by the policy. 
 	The information flow graph captures the counterintuitive behaviour arising from the interplay of wildcard characters and variables, which are resolved at different times during request evaluation and may lead to \emph{wildcard injection attacks}.
 	Checking an information flow thus reduces to reachability in this graph.
    We compute a finite symbolic version of the graph through an SMT solver to enable practical verification, leveraging the specific constraints of AWS policies
		to efficiently decide the satisfiability of formulas with regular languages and string manipulation.
	\item Finally, we implement our verification mechanism in the tool \thetool, available online~\cite{repo}.
   We evaluate the tool's effectiveness in detecting misconfigurations and unintended information flows using a case study designed to capture a typical Build Automation System. Although synthetic, the case study reflects the main structural constraints and operational patterns of real IoT systems.  
   In addition, we assess the tool's performance by randomly synthesising a network of devices and associating them with real-world IoT policies available online~\cite{p-verifier-website}: we show that \thetool\ scales well in practice as the network grows.  
\end{enumerate}

\subsubsection*{Plan of the paper}
In \Cref{sec:background}, we introduce the MQTT protocol and AWS IoT Core.
\Cref{sec:case:study} present our running example.
\Cref{sec:model} formalises AWS IoT Core's main components.
In \Cref{sec:infoflow}, we describe how to build the information flow graph from a set of policies.
\Cref{sec:tool} presents our tool and its experimental evaluation, 
\Cref{sec:ext:model} discusses the assumptions and limitations of our model, and how to extend our proposal to additional mechanisms available in IAM policies.
In \Cref{sec:ralated,sec:concl}, we compare our approach with the literature and draw some conclusions.
Appendices contain a summary of symbols and notation (\Cref{tab:notation}), a formal model of device-broker interaction, the proofs of our formal development, and some complementary evaluation results.

\section{Background}\label{sec:background}

Here, we describe the MQTT protocol and the AWS IoT policy language by focusing
only on the main components, which are sufficient for our formal development.
Possible extensions and advanced features are discussed in~\Cref{sec:ext:model}.
We also partially adapt the syntax for brevity, e.g. by omitting the Amazon Resource Name (ARN) from IoT-AWS policies.

\subsubsection*{MQTT Protocol}\label{sec:MQTT}
Message Queuing Telemetry Transport (MQTT)~\cite{MQTTSpec} is a lightweight client/server protocol relying on a publish/subscribe communication model, designed for environments where resources and bandwidth are constrained and limited, such as 
machine-to-machine
and IoT contexts.
In the publish/subscribe communication model, the entities sending messages (\emph{the publishers}) and the ones
receiving them (\emph{the subscribers}) indirectly interact through the infrastructure 
provided by a third component (\emph{the broker}), which is the only one responsible for delivering messages.
In more detail, 
clients send $\code{CONNECT}$ requests to establish a connection with the broker, specifying their chosen client id.
After connection, clients communicate via \emph{topics:}
virtual communication channels managed by the broker.
In practice, 
topics are
hierarchically-structured strings,
organized into \emph{topic levels} through forward-slash characters `$\code{/}$'.
For example, the topic $\code{floor1/room1/temp}$ can be used to communicate 
the sampled temperature in the given room and floor of a building. 
The broker performs pattern-matching to filter and forward incoming messages to interested clients.
A client can subscribe to a given set of topics 
by sending a $\code{SUBSCRIBE}$ request to the broker,
specifying a \emph{topic filter}, namely a restricted regular expression denoting the set of topics it is interested in.
Topic filters can use the single-level `+' or the multi-level `$\code{\#}$' wildcard characters:
`+' replaces any single topic level, whereas `$\#$' at the end of a topic filter represents any sequence of topic levels.
For example, the topic filter $\code{floor1/{+}/temp}$ allows a client to receive temperature values from any room on the first floor of a building.
For publishing on a given topic, 
clients send a $\code{PUBLISH}$ message
specifying the topic and a payload; the broker then propagates the payload to all clients subscribed to that topic.
For example, a temperature sensor in $\code{room1}$ can publish on the topic $\code{floor1/room1/temp}$ to inform all the clients interested in that temperature value.

\subsubsection*{AWS IoT Policies}
Amazon Web Services (AWS)~\cite{AWS} provides cloud services
and a communication infrastructure on which IoT developers can deploy their devices, known as AWS IoT Core~\cite{IoTCore}.
Among the various technologies provided by AWS IoT Core, we find the MQTT protocol: the cloud acts as 
the MQTT broker, while providing security mechanisms via the implementation of authentication and access control.
Devices may connect to the broker only if they successfully authenticate
through a certificate or via other Amazon authentication services~\cite{IAM,Cognito};
moreover, even when a client is authenticated, access to the MQTT infrastructure is controlled 
via a collection of access control policies, which 
restricts the available topics that any given client may use.

IoT Core allows the definition of different access rights,
which correspond to MQTT actions.
Here, we only consider the 
main ones:
\begin{enumerate*}[label=(\roman*)]
  \item $\code{iot:Connect}$ to connect to the cloud infrastructure;
  \item $\code{iot:Publish}$ to publish on a certain topic;
  \item $\code{iot:Subscribe}$ to subscribe to some topic filter;
  \item $\code{iot:Receive}$ to receive the messages
  published on the specified topics.
\end{enumerate*}
Note that clients need both the $\code{iot:Receive}$ and $\code{iot:Subscribe}$ access rights to receive messages on a given topic.
The access rights are specified in a policy (a JSON 
document) consisting of a set of \emph{policy statements} (access control rules). 
Each statement describes whether a specific action, e.g.~$\code{iot:Publish}$, is permitted or not ($\allow$ or $\deny$) on a given resource (topic, topic filter, client id).
Policy developers can use wildcard characters `$?$' and `$\ast$' to represent any single character and any sequence of characters, respectively.
For example, the following policy:
\begin{lstlisting}[basicstyle=\ttfamily\footnotesize]
{ "Effect": "Allow", 
  "Action": "iot:Subscribe",
  "Resource": [ "floor1/*/temp", "floor1/room1/*" ]},
{ "Effect": "Allow", 
  "Action": "iot:Receive",
  "Resource": [ "floor1/*/temp", "floor1/room1/*" ]}
\end{lstlisting}
grants a client permission to subscribe and receive messages on topics concerning the temperature of any room on the first floor and any sensor in $\code{room1}$ of the first floor.
Moreover, policy statements may use variables, such as $\code{\$\{iot:ClientId\}}$, 
which are substituted for attributes associated with the certificate provided during connection.
Permissions are then evaluated on the resulting variable-free policy by following a \emph{default-deny} strategy: access is granted if there is at least a statement
allowing the requested action and no statement denying it; in any other case, access is rejected.

\begin{figure*}[t]
\centering
\scalebox{1.1}[1.1]
{
\begin{tikzpicture}[
	node distance=3mm,
	arrow/.style={-stealth,thick},
	obj/.style={font=\footnotesize},
	obj/.style={font=\footnotesize},
	weight/.style={font=\footnotesize,midway},
	title/.style={font=\fontsize{7}{7}\color{black!80}\ttfamily},
    typetag/.style={rectangle, draw=black!50, font=\footnotesize\ttfamily, anchor=west},
    topic/.style={rectangle, thick, draw=black, font=\footnotesize\ttfamily}
]
\node [title] (phy) at (0,0) {physical AC system for floor 1/2};
\node [obj,  above left=1mm and -12mm of phy] (bulb) {
	\includegraphics[width=5mm,keepaspectratio]{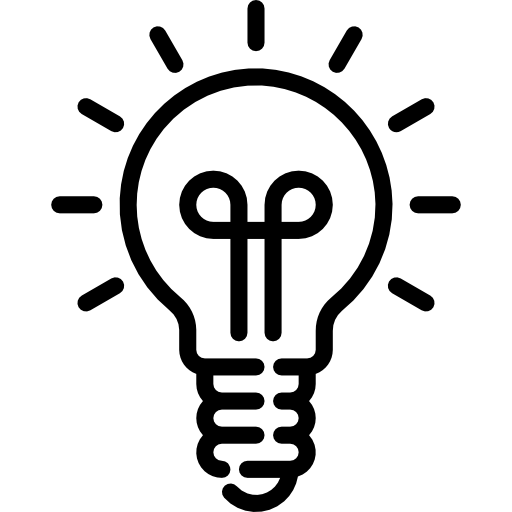}
};
\node [obj, right= of bulb] (psens) {
	\includegraphics[width=5mm,keepaspectratio]{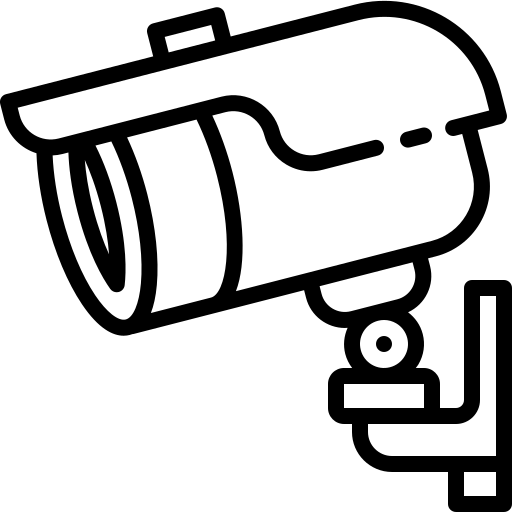}
};
\node [obj, right= of psens] (lock) {
	\includegraphics[width=5mm,keepaspectratio]{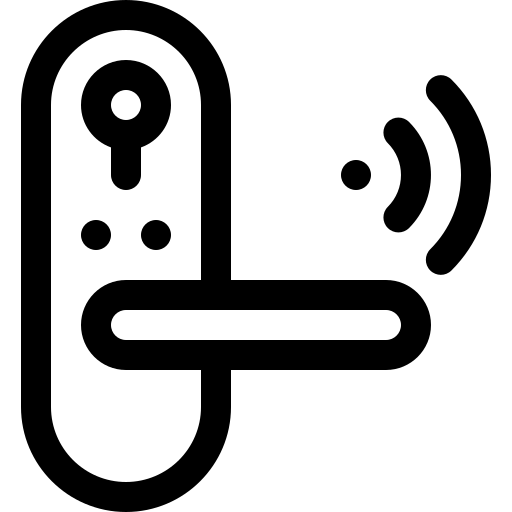}
};
\node [obj, right= of lock] (bread) {
	\includegraphics[width=5mm,keepaspectratio]{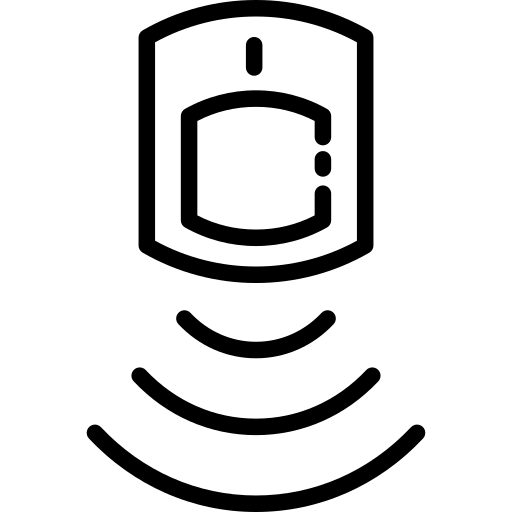}
};

\begin{scope}[on background layer]
\node [fill=black!10, inner ysep=.5mm, inner xsep=1.5mm, fit={(phy) (bread) (psens) (lock) (bulb)}] {};
\end{scope}

\node [obj, right=17.5mm of bread] (pump) {
	\includegraphics[width=5mm,keepaspectratio]{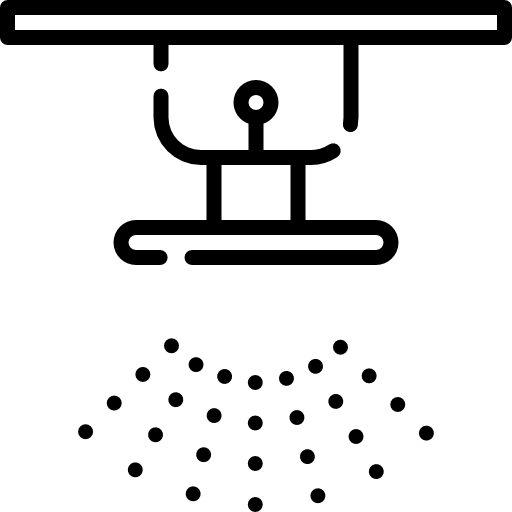}
};
\node [obj, right= of pump] (aalarm) {
	\includegraphics[width=5mm,keepaspectratio]{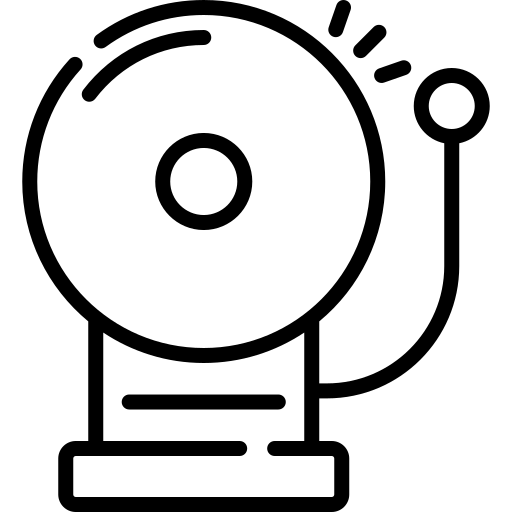}
};
\node [obj, right= of aalarm] (abutton) {
	\includegraphics[width=5mm,keepaspectratio]{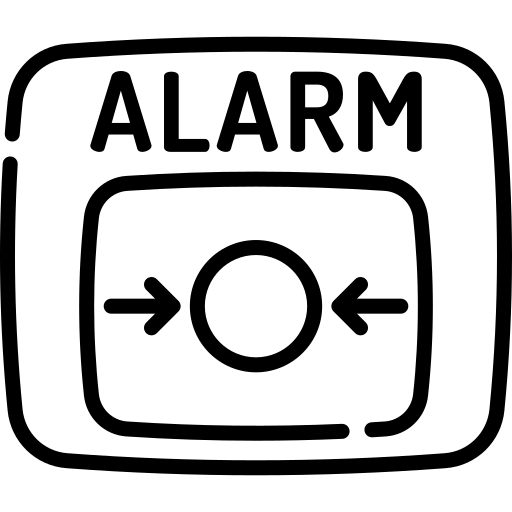}
};
\node [obj, right= of abutton] (ssens) {
	\includegraphics[width=5mm,keepaspectratio]{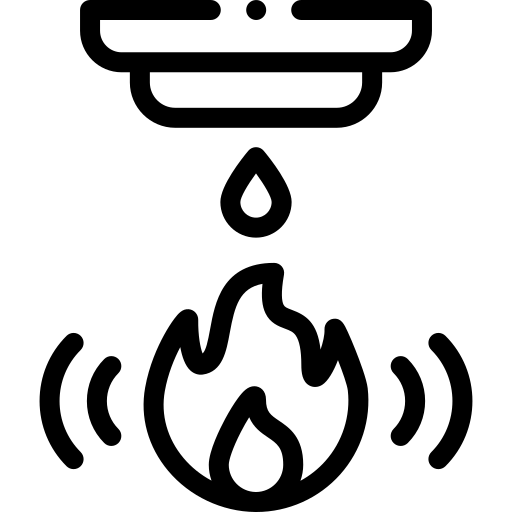}
};
\node [title, below right=1mm and -11mm of pump] (fire) {fire alarm system for floor 1/2};
\begin{scope}[on background layer]
\node [fill=black!10,inner ysep=.5mm, inner xsep=1.5mm,fit={(fire) (ssens) (pump) (aalarm) (abutton)}] {};
\end{scope}

\node [obj, right=14mm of ssens] (elevator) {
	\includegraphics[width=5mm,keepaspectratio]{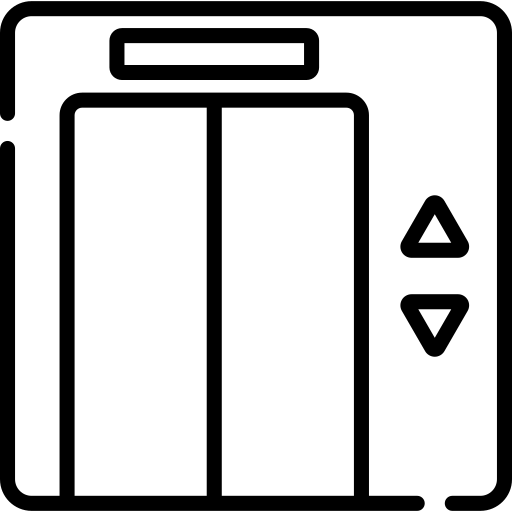}
};
\node [title, below right=1mm and -9.5mm of elevator] (other) {\phantom{y}other\phantom{/}};
\begin{scope}[on background layer]
\node [fill=black!10,inner ysep=.5mm, inner xsep=0mm,fit={(other) (elevator)}] {};
\end{scope}

\node [fit={(elevator) (phy)}] (fl) {};

\node [obj, left=5mm of fl] (fltytle) {
	\footnotesize
	\begin{tabular}{c}
	Field\\[-1mm]
	Layer
	\end{tabular}
};
\node [obj, above=5mm of fltytle] (tltytle) {\footnotesize Topics};
\node [obj, above=1.5mm of tltytle] (cltytle) {
	\begin{tabular}{c}
	\footnotesize
	Control\\[-1mm]
	Layer
	\end{tabular}
};

\node [topic, above=5mm of $(psens.north)!0.5!(bulb.north)$] (t1) {};
\node [topic, above=5mm of $(psens.north)!0.5!(lock.north)$] (t2) {};
\node [topic, above=5mm of $(lock.north)!0.25!(bread.north)$] (t3) {};
\node [topic, above=5mm of $(lock.north)!0.75!(bread.north)$] (t4) {};
\node [topic, above=5mm of $(pump.north)!0.5!(aalarm.north)$] (t5) {};
\node [topic, above=5mm of $(ssens.north)!0.5!(abutton.north)$] (t6) {};

\node [obj, inner ysep=0mm, inner xsep=0mm, above=4mm of $(t1.north)!0.5!(t2.north)$] (serv1) {
	\includegraphics[width=7mm,keepaspectratio]{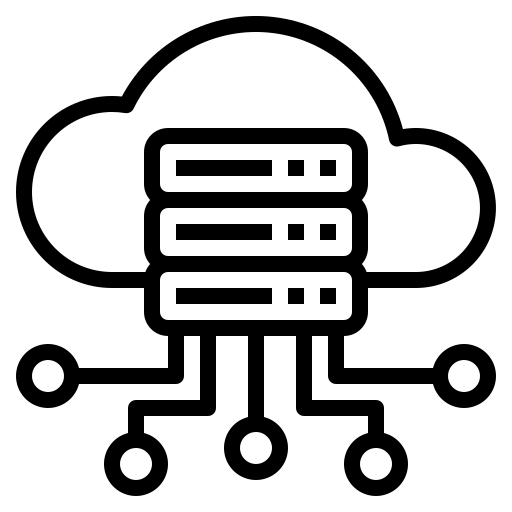}
};
\node [obj, inner ysep=.5mm, inner xsep=1mm, fill=white, fill opacity=0.85,text opacity=1, text=black] at (serv1) {
	$\codea{log}$
};

\node [obj, inner ysep=0mm, inner xsep=0mm, above=4mm of $(t3.north)!0.5!(t4.north)$] (serv2) {
	\includegraphics[width=7mm,keepaspectratio]{img/cloud}	};
\node [obj, inner ysep=.5mm, inner xsep=1mm, fill=white, fill opacity=0.85,text opacity=1, text=black] at (serv2) {
	 $\codea{AC list}$
};

\node [obj, inner ysep=0mm, inner xsep=0mm, above=4mm of $(t5.north)!0.5!(t6.north)$] (serv3) {
	\includegraphics[width=7mm,keepaspectratio]{img/cloud}
};
\node [obj, inner ysep=.5mm, inner xsep=1mm, fill=white, fill opacity=0.85,text opacity=1, text=black] at (serv3) {
	$\codea{fire manager}$
};

\draw [arrow] (t1) -- (bulb);
\draw [arrow] (psens) -- (t1);
\draw [arrow] (t2) -- (psens);
\draw [arrow] (lock) -- (t2);
\draw [arrow] (t3) -- (lock.north);
\draw [arrow] (bread) -- (t4);
\draw [arrow] (t5) -- (lock.north);
\draw [arrow] (t5) -- (pump.north);
\draw [arrow] (t5) -- (aalarm.north);
\draw [arrow] (t5) -- (elevator.north);
\draw [arrow] (abutton.north) -- (t5);
\draw [arrow] (ssens.north) -- (t6);

\draw [arrow] (t1) -- (serv1);
\draw [arrow] (t2) -- (serv1);
\draw [arrow] (serv2) -- (t3);
\draw [arrow] (t4) -- (serv2);
\draw [arrow] (serv3) -- (t5);
\draw [arrow] (t6) -- (serv3);

\end{tikzpicture}}
\caption{The Architecture of Our Building Automation System.}
\label{fig:bas}
\end{figure*}

\section{A running example: a Smart Building}
\label{sec:case:study}

Assume a smart building with two floors, where several subsystems, including energy consumption, physical access control, heating, and fire alert, are all integrated.
Our Building Automation System (BAS) is cloud-based, and it is organised in two layers (see \Cref{fig:bas}):
the field layer contains off-the-shelf IoT devices that interact with the physical world via sensors and actuators, while the control layer consists of an MQTT broker and several cloud services that implement the communication infrastructure and the control logic.
These services execute appropriate actions when triggered by events occurring in the field layer and provide operators with dashboards to monitor, configure, and control the system.
Moreover, we group IoT devices in subsystems according to their functionalities.
For example, smoke detectors are part of the fire alarm system, whereas badge readers are part of the physical access control system. 

The physical access control system includes badge readers, smart locks, presence sensors, and light bulbs.
Each floor of the building is expected to behave independently.
Consider the first floor as an example.
When the badge reader scans a bar code, it publishes a message on the topic $\code{phAC/floor1/bdgReader1/check}$.
The cloud service $\code{AClist}$ waits for messages on this topic and checks the code's validity in the message's payload.
If the badge is authorised to open the door, $\code{AClist}$ publishes a message on the topic $\code{phAC/floor1/lock1/open}$. 
When the smart lock $\code{lock1}$ receives such a message, it unlocks itself and activates the presence sensor of the room by publishing on the topic $\code{phAC/floor1/prsSens1/enable}$.
When the presence sensor $\code{prsSens1}$ is active and detects someone inside the room, it publishes on the topic $\code{phAC/floor1/dtdMovement/light1}$. 
The light bulb $\code{light1}$ is subscribed to this topic and switches itself on 
upon reception.
The service $\code{log}$ subscribes to all the topics mentioned above to keep a trace of 
such events.

The fire alarm subsystem consists of smoke sensors, alarm buttons, acoustic alarms, and water pumps.
Its expected behaviour for the first floor follows.
When a smoke sensor detects some smoke in its area, it publishes a message on the topic $\code{fire/floor1/smokeLvl}$ that is received by the cloud service $\code{fireMngr}$ subscribed to the related topic filter.
If the level of detected smoke exceeds a given threshold, $\code{fireMngr}$ publishes on the topic $\code{fire/detected}$ to raise the fire alarm. 
The water pumps, acoustic alarms, doors, and elevators all subscribe to this topic.
When they receive a message on it, the water pumps and acoustic alarms activate, the doors unlock, and the elevator travels to the ground floor, opens its doors, and halts.
The alarm can also be raised by the fire alarm button, which publishes to the same topic.%

The IoT policies for each device of the BAS system are available online~\cite{repo}.
We assume that the cloud infrastructure is trusted, while devices may be compromised and policies misconfigured.
An attacker can gain full control of vulnerable devices (e.g., a light bulb) and use them to compromise other critical devices and obtain private information.
The following sections present our formal model, our verification framework and the tool \thetool, which detects policy misconfiguration and unintended information flows.

\section{A Formal Model of IoT Access Policies}\label{sec:model}

We take AWS IoT Core as a reference implementation of an IoT-Cloud infrastructure, 
and we formalise a core of the AWS policy language and its semantics, focusing on the interplay between MQTT and the authorisation mechanism.
In doing so, we clarify the counterintuitive semantics of AWS and MQTT wildcards and variable substitution in policy evaluation.

Our presentation follows a bottom-up approach.
\begin{definition}[Resource]
A resource $\rho$ is either a client id, a topic, or a topic filter.
Let $\Lambda$ be the alphabet of alphanumerical characters enriched with the forward-slash character `$/$'.
Topics are strings over $\Lambda$, while client ids and topic filters are strings over 
$\Lambda \cup \{ \code{+}, \code{\#} \}$.
\end{definition}

We denote with $\cids$, $\topics$, and $\topicfs$ the set of client ids, topics, and topic filters, respectively.
Note that these sets have a non-empty intersection, so as to model
features of the language like substitutions of client ids for $\cidVAR$ variable and matching topics with topic filters.
In the following, we see that MQTT wildcards are treated as common characters in client ids, and as special characters when occurring in topic filters, which allows \emph{wildcard injection}.
Indeed, although client ids and topic filters have the same syntax, a topic filter $\tf$ can be interpreted 
as a regular expression $\tf^{\mqtt}$ where the character `$\code{\#}$' is interpreted as any sequence of characters in $\Lambda$, while the character `$\code{+}$' as any sequence of characters different from `$\code{/}$' in $\Lambda$.
According to MQTT rules, the substitution takes place only if `$\code{\#}$' is the terminal character and `$\code{\#}$' and `$\code{+}$' appear as topic levels, i.e. they are separated from other characters by `$\code{/}$'.
The language of $\tf^{\mqtt}$, denoted as $\lang{\tf^{\mqtt}}$, contains the topics that are matched by the topic filter $\tf$.
\begin{example}
The topic filter $\code{floor1/{+}/temp}$ matches the topic $\code{floor1/room10/temp}$ because it is recognized by the regular expression 
$\code{floor1/[a{-}zA{-}Z0{-}9]\Ast/temp}$, where
$\code{[a{-}zA{-}Z0{-}9]}$ is any alphanumerical character.
\end{example}

In a policy, resources are specified by strings called resource expressions, which
may contain 
AWS wildcards to compactly represent a set of resources, and the variable $\code{\$\{iot:ClientId\}}$ (that we 
abbreviate as
$\cidVAR$).
\begin{definition}[Resource Expression]
A resource expression $\re$ for 
topics, topic filters, or client ids,
is a string built on the same alphabet of the 
given kind of resource extended with the variable $\cidVAR$ and with the AWS wildcards \emph{`$\code{?}$'} and \emph{`$\code{\ast}$'}.
A resource expression is \emph{ground} if it does not contain $\cidVAR$.
\end{definition}
Given a resource expression $\re$ and a string $v$, we 
write
$\re[\sfrac{v}{\cidVARa}]$ for
the resource expression where every occurrence of $\cidVAR$ is replaced with $v$.
For example, $\code{dtdMovement?/\cidVAR}[\sfrac{\mathtt{\#}}{\cidVARa}]$ is $\code{dtdMovement?/\#}$.
A ground resource expression $\rExp$ 
can be interpreted as a regular expression $\rExp^{\AWS}$ where `$?$' 
stands for
any character in $\Lambda$ extended with the MQTT wildcards, and `$\ast$' 
for
any sequence of such characters. %
The language $\lang{\rExp^{\AWS}}$ contains the resources that are matched by the resource expression $\rExp$.
\begin{example}\label{ex:rexp}
After the device connects with client id $\code{light1}$, the topic filter $\tf = \code{phyAC/floor1/dtdMovement/light1}$ is a resource matched by 
the expression $\rExp =\code{phyAC/floor?/dtdMovement/\cidVAR}$, which uses the AWS wildcard `$?$'.
Indeed, $\tf \in \lang{\rExp[\sfrac{\mathtt{light1}}{\cidVARa}]^{\AWS}}$ since $\rExp[\sfrac{\mathtt{light1}}{\cidVARa}]$ can be seen as the regular expression 
$\code{phyAC/floor[a{-}zA{-}Z0{-}9/{+}\#]/dtdMovement/light1}$.
\end{example}

We now define IoT policies: a policy establishes which client ids can connect and which topics they can publish and/or subscribe to.
\begin{definition}[Policy]\label{def:pol}
A policy $\pol$ is a set of statements $s = \polState$ 
where
  \begin{itemize}
    \item $\effect \in \{\allow, \deny\}$ 
          is the effect of the statement; 
    \item $\action \in \{\iotConn, \iotPub, \iotRec, \iotSub \}$ 
    	is an IoT action;
    \item $\RE$ is a set of resource expressions.
  \end{itemize}
  The definition of ground and of substitution is naturally extended to policies.
\end{definition}

\begin{figure}[t]
\begin{lstlisting}[basicstyle=\ttfamily\footnotesize]
{ "Effect":"Allow", 
  "Action":"iot:Connect", 
  "Resource":"*" },
{ "Effect":"Allow", 
  "Action":"iot:Receive", 
  "Resource":"*" },
{ "Effect":"Allow", 
  "Action":"iot:Subscribe",
  "Resource":
      "phyAC/floor?/dtdMovement/${iot:ClientId}" }
\end{lstlisting}
\caption{The policy of the light bulbs in the BAS.}
\label{fig:light:bulb:policy}
\end{figure}

\begin{example}
\Cref{fig:light:bulb:policy} shows the policy of a light bulb in our running example, using the resource expression of~\Cref{ex:rexp}.
\end{example}

To configure IoT Core, developers define certificates to identify the devices authorised to connect to the broker.
Devices are associated with their certificates, and 
each certificate is associated with its policies. 
Formally:
\begin{definition}[Configuration]\label{def:iot:config}
A configuration is a 5-tuple $(\certs, \pols, \varphi, \Dev, k)$ where
  \begin{itemize}
    \item $\certs$ is a set of certificates;
    \item $\pols$ is a set of IoT policies;
    \item $\varphi\colon \certs \rightarrow \mathcal{P}(\pols)$ 
    associates certificates with their policies;
    \item $\Dev$ is a set of devices;
    \item $k\colon \Dev \rightarrow \mathcal{P}(\certs)$ associates devices with their certificates.
  \end{itemize}
\end{definition}
We write $\Phi(c) = \bigcup_{\pol \in \varphi(c)} \pol$ for the union of $c$'s policies.

\begin{example}
Consider the BAS of \Cref{sec:case:study}, the certificates of both the light bulbs are associated with the policy of~\Cref{fig:light:bulb:policy},
i.e., $\varphi(c_{\codea{light1}}) = \varphi(c_{\codea{light2}})$.
\end{example}

We now introduce how permission requests are evaluated.
Intuitively, the broker interrogates the access control system to check if the action requested by the message is allowed.
To do that, the access control system checks whether there exists a policy statement $s$ in $\pol$ forbidding the request, 
i.e., with the same action of the request, with effect $\deny$, and with the resource $\res$ matched by a resource expression in $\RE$ (after expanding AWS wildcards). 
If a match is found, the request is rejected; otherwise, the system looks for a statement $s$ that explicitly allows the request. 
If there is at least one match, the request is accepted.
If no matching statement is found, the request is denied (\emph{implicit deny}).

\begin{definition}[Permission]\label{def:perm}
A permission is a pair $(\action, \rho)$ where $\alpha$ is an IoT action to be performed over a resource $\rho$.
A permission $(\action, \rho)$ is granted by the \emph{ground} policy $\pol$, in symbols 
$\pol \permit (\action, \rho)$, if and only if both the following hold:
\begin{itemize}
\item $\{ (\allow, \alpha, \{\re\} \cup \RE) \}$ exists in $\pol$ with
$\rho \in \lang{\re^{\AWS}}$;
\item 
for all
$(\deny, \alpha, \RE) \in \pol$ and $\re \in \RE$, $\rho \notin \lang{\re^{\AWS}}$.
\end{itemize}
\end{definition}

We can now detail how an attacker can break desired security properties by exploiting policy misconfigurations.
In more detail, we assume that the cloud provider (thus, the broker and authorisation mechanisms) is trusted and safe from tampering. 
On the contrary, devices may be compromised by attackers, and IoT policies may be misconfigured and subverted.
In addition, we assume that the attacker knows the configuration of the system, and it can take full control of a device $\dev$, thus gaining the capability of authenticating with all the certificates in $k(\dev)$.
The compromised device $\dev$ ignores its original scope,
and it maliciously selects both the client id to connect with and the topics over which to communicate.
The attacker's objective is to compromise either \emph{confidentiality} or \emph{integrity} of the system.
When breaking confidentiality, the attacker gains private information, e.g.~directly from some sensor $\dev$, or indirectly from other devices $\dev'$ receiving messages from $\dev$.
Conversely, the integrity of a device $\dev$ is compromised when its behaviour is influenced by the attacker, possibly leveraging an intermediate device that propagates the attacker's messages.

We conclude by clarifying how the \emph{wildcard injection attack} works.
MQTT wildcards are not considered as reserved characters when used in client ids, but they are substituted verbatim for $\cidVAR$.
During connection, attackers can use a string containing a wildcard as their client id, thus injecting it 
into the policies of the broker, possibly bypassing $\deny$ rules and breaking security guarantees.
For example, consider a certificate $c$ and a device $d$ with $c \in k(d)$, and assume that:
\begin{align*}
\Phi(c) &= \{(\allow, \iotConn, \{*\}), (\deny, \iotConn, \{\code{private}\}), \\
    & (\allow, \iotRec, \{*\}), (\allow, \iotSub, \{/\cidVAR\})\}
\end{align*}
The policy should allow every device to receive on its topic (named after its client id), while forbidding them from using the topic $/\code{private}$, which is reserved for protected information.
At first glance, eavesdropping seems impossible.
However, assume $d$ connects to the broker by identifying as $\#$.
Since $\code{\#} \in \lang{\ast^{\AWS}}$, the connection permission is granted, and the   
last statement of $\Phi(c)$ is instantiated as $(\allow, \iotSub, \{/\#\})$.
Note that the instantiated policy allows $d$ to subscribe to the topic filter $\code{/\#}$, where the character `$\code{\#}$' of the client id is interpreted as an MQTT wildcard.
Thus, the device can receive all the messages published on the topic $\code{/private}$.

In the next section, we make this formal by giving a graph representation of the possible paths along which messages can propagate through devices and topics.

\section{Characterising Information Flow}\label{sec:infoflow}

We characterise the information flow permitted by an IoT configuration, i.e. which devices can influence each other.
The underlying idea is to build a graph whose nodes represent devices
and topics: 
there is an arc from a device $\dev$ to a topic $t$ if a policy associated with the device allows it
to publish on $t$; 
vice versa there is an arc from $t$ to $\dev$ if a policy allows $\dev$ to subscribe and receive messages on $t$.    
Then, we present a symbolic representation of this graph, enabling the practical verification of security properties.

\subsection{Information Flow Graph}

It is convenient to introduce some auxiliary definitions to facilitate the construction of our graphs. 
For each certificate $c$, 
we define $\Sin_c$ as the set of topics from which the devices connected and authenticated using certificate $c$ are permitted to receive messages.
Similarly, $\Sout_c$ collects topics to which $c$ may publish messages.
\begin{definition}[I/O Sets] \label{def:ioset}
Given a configuration, the 
input and output sets of certificate $c$ are the sets of topics $\Sin_{c}$ and $\Sout_{c}$:
 \begin{align*}
    \Sin_{c} = \{ t \mid &\
				\Phi(c)[\sfrac{\cid}{\cidVARa}] \permit (\iotConn, \cid),\\
				&\ \Phi(c)[\sfrac{\cid}{\cidVARa}] \permit (\iotRec, t),\\
				&\ \Phi(c)[\sfrac{\cid}{\cidVARa}] \permit (\iotSub, \tf)\\
				&\text{ and } t \in \lang{\mathit{tf}^{\mqtt}}, \text{ for some } \cid \text{ and }\tf\}\\
    \Sout_{c} = \{ t \mid & \
				\Phi(c)[\sfrac{\cid}{\cidVARa}] \permit (\iotConn, \cid)\\
				&\text{ and } \Phi(c)[\sfrac{\cid}{\cidVARa}] \permit (\iotPub, t), \text{ for some } \cid \}
 \end{align*}
\end{definition}
Roughly speaking, the conditions on the first set require $c$ to be able to connect with client id $\cid$, to subscribe to the topic filter $\mathit{tf}$, and that topic $t$ belongs to the language of  $\mathit{tf}$.
The conditions on the second set, instead, require that $c$ can connect with client id $\cid$ and can publish over $t$.

\begin{example}\label{ex:fin:fout}
Consider the light bulb $\code{light2}$ and the topic $t = \code{phAC/floor1/dtdMovement/light1}$.
It holds that $t \in \Sin_{c_{\codea{light2}}}$ because: 
$(i)$ the light bulb can connect to the broker with client id $\#$, and the substitution gives the policy 
\begin{align*}
     P' = \{&(\allow, \iotConn, \{*\}), (\allow, \iotRec, \{*\}),\\
    &(\allow, \iotSub, \{\code{phAC/floor?/dtdMovement/\#}\})\}
\end{align*}
$(ii)$ the ground policy allows the client to receive messages on $t$ and to subscribe to the topic filter $\tf = \code{phAC/floor1/dtdMovement/\#}$; and $(iii)$ the topic $t$ is a string in the language of $\mathit{tf}$ when considering MQTT wildcards.
Note that this is a wildcard injection attack: the second light bulb receives private information intended for another device.
\end{example}

We now define information flow graphs of configurations:

\begin{definition}[Information Flow Graph]\label{def:info-graph}
  Given a configuration, its information flow graph is a bipartite directed graph $G = (\Dev \cup \topics, E)$ with
  $\Dev$ the set of devices,
  $\topics$ of all possible topics,
  and 
  $E \subseteq (\Dev \times \topics) \cup (\topics \times \Dev)$ of the arcs defined as
  \begin{align*}
  E = \{
  (d,t) \mid \exists c \in k(d) \text{ s.t. } t \in \Sout_{c}
  \}\\
  \cup\ \{
  (t,d) \mid \exists c \in k(d) \text{ s.t. } t \in \Sin_{c}
  \} 
  \end{align*}   
\end{definition}

An arc of $G$ represents a communication permitted by the configuration.
Following a path of $G$, we can track all the possible device interactions, including indirect ones.
We say that there is a (possible) information flow from $d$ to $d'$
if $d'$ is reachable from $d$ in the information flow graph,
meaning that information produced by $d$ may affect $d'$.
\begin{example}\label{ex:violating:path}
Assume that you want 
the information on $\code{lock1}$ to be private to the physical access control of the first floor in our running example.
This requirement can be checked by inspecting the information flow graph of the configuration in~\cite{repo}.
Indeed, the following path is a violation ${\code{lock1}} \rightarrow t_1 \rightarrow {\code{prsSens1}} \rightarrow t_2 \rightarrow {\code{light2}}$ 
where $t_1 = \code{phAC/floor1/prsSens1/enable}$ and $t_2 = \code{phAC/floor1/dtdMovement/light1}$.
\end{example}

As a final remark, recall that we do \emph{not} identify a client with its client id.
This is because only the association with the certificate is authenticated by the Cloud, whereas
the client id is freely selected by the device.
A compromised device might exploit this flexibility by using multiple ids or MQTT wildcards to maximise the attack surface.  

In~\Cref{sec:broker:sem}, we assess the adequacy of the information flow graph in capturing how information propagates in IoT systems.
We give an operational semantics describing the effects of MQTT requests on message transmission between devices and on the broker state.
Then, we characterise information flow as sequences of messages that may arise in some feasible execution, and we prove it consistent with the paths of the information flow graph built from the configuration.

\subsection{Symbolic Information Flow Graph}

In general, the number of nodes and arcs of the information flow graph grows exponentially with respect to the length of topics 
(256 bytes of UTF-8 characters in the current implementation of AWS IoT Core).
Clearly, building such a graph is impractical. 
To address this issue, we build a symbolic version where sets of topics are represented by logical formulas.

We define predicates representing the conditions under which a permission may be granted during some execution.
\begin{definition}[Permission Predicate] \label{def:eval}
 Given a 4-tuple $\req$ where $c$ is a certificate, $\action$ is an IoT action, and $\res$ is a resource, we 
 let the \emph{permission predicate} $\reqEval{\req}$ be defined as follows, where $\pol$ is $\Phi(c)[\sfrac{\cid}{\cidVARa}]$: 
\begin{align*}
    \reqEval{(c, \cid, \action, \res)} := 
      &\Bigl(
        \bigvee_{(\allowa, \action, \RE) \in \pol}
        \bigvee_{\re \in \RE}
        \res \in \lang{\re^{\AWS}}
      \Bigr) \land\\
    &\lnot
      \Bigl(
        \bigvee_{(\denya, \action, \RE) \in \pol}
        \bigvee_{\re \in \RE}
        \res \in \lang{\re^{\AWS}}
      \Bigr)
\end{align*} 
\end{definition}

We introduce the auxiliary predicates $\Fin_c(t)$ and $\Fout_c(t)$, mimicking input and output sets in a symbolic setting.
\begin{definition}[I/O Predicates] \label{def:iopred}
Given a configuration, the input and output predicates $\Fin_c$ and $\Fout_c$ of a certificate $c$ are:
 \begin{align*}
    \Fin_{c}(t) &= \exists \cid, \mathit{tf} .
          \reqEval{(c, \cid, \iotSub, \mathit{tf})} \land 
          \reqEval{(c, \cid, \iotRec, t)}\ \land\\
          &\qquad\reqEval{(c, \cid, \iotConn, \cid)} \land t \in \lang{\mathit{tf}^{\mqtt}};\\
    \Fout_{c}(t) &= \exists \cid . \reqEval{(c, \cid, \iotPub, t)} \land \reqEval{(c, \cid, \iotConn, \cid)}
 \end{align*}
\end{definition}

Intuitively, the symbolic version of the information flow graph has a node for each device, and an additional node for each pair of certificates that can communicate.
More precisely, for each pair of certificates $c$ and $c'$, we build a logic formula $F_{c,c'}$ that is true if and only if  $c$ can publish a message that can be read by (devices authenticated with) $c'$.
If the formula is valid, we add a node $F_{c,c'}$, and we connect it with an arc from $d$ to $F_{c,c'}$ 
for each device such that $c \in k(d)$.
Similarly, we add an arc from $F_{c,c'}$ to $d'$ if $c' \in k(d')$.
(see the next section for the algorithm to compute the symbolic graph). 
\begin{definition}[Symbolic Information Flow Graph]\label{def:info-graph:sym}
The symbolic information flow graph $G_s$ of a 
	configuration is $(N, E)$:
  \begin{align*}
    N &= \Dev \cup \{ F_{c,c'} \, \mid c,c' \in \certs\, \text{ and }\, F_{c,c'} \, \text{is true} \}\\
    E &= \{(d,F_{c,c'}) \mid c \in k(d) \} \cup \{(F_{c,c'},d') \mid c' \in k(d')\}  
  \end{align*}
  where for each $c$ and $c'$, $F_{c, c'} = \exists t . \Fout_c(t) \land \Fin_{c'}(t)$ is a logic formula
  over the (first-order) theory of regular expressions.

\end{definition}
We say that there is a (possible) symbolic information flow from $d$ to $d'$
if $d'$ is reachable from $d$ in the symbolic information flow graph $G_s$.

\begin{figure}[t]
\centering

\scalebox{0.625}[0.625]{
\begin{tikzpicture}[
	xscale=1.5,yscale=.95,
	node distance=3mm,
	arrow/.style={-stealth,lightgray},
	obj/.style={font=\normalsize\ttfamily, rounded rectangle, draw=black!50},
	weight/.style={font=\footnotesize,midway},
	title/.style={font=\fontsize{6}{6}\color{black!50}\ttfamily},
    typetag/.style={rectangle, draw=black!50, font=\scriptsize\ttfamily, anchor=west},
    topic/.style={rectangle, draw=black, thick, font=\scriptsize\ttfamily}
]

\node[obj] (light1) at (36:4cm) {$\codea{light1}$};
\node[obj] (light2) at (54:4cm) {$\codea{light2}$};

\node[obj] (abut1) at (0:4cm) {$\codea{button1}$};
\node[obj] (abut2) at (18:4cm) {$\codea{button2}$};

\node[obj] (psens1) at (252:4cm) {$\codea{prsSens1}$};
\node[obj] (psens2) at (270:3.8cm) {$\codea{prsSens2}$};

\node[obj] (ssens1) at (216:4cm) {$\codea{smoke1}$};
\node[obj] (ssens2) at (234:4cm) {$\codea{smoke2}$};

\node[obj] (lock1) at (144:4cm) {$\codea{lock1}$};
\node[obj] (lock2) at (162:4cm) {$\codea{lock2}$};
\node[obj] (pump1) at (72:4cm) {$\codea{pump1}$};
\node[obj] (pump2) at (90:3.8cm) {$\codea{pump2}$};

\node[obj] (log) at (108:4cm) {$\codea{log}$};
\node[obj] (AClist) at (126:4cm) {$\codea{AClist}$};

\node[obj] (aalarm1) at (180:4cm) {$\codea{acscAlarm1}$};
\node[obj] (aalarm2) at (198:4cm) {$\codea{acscAlarm2}$};

\node[obj] (elev) at (288:4cm) {$\codea{elevator}$};

\node[obj] (bread1) at (324:4cm) {$\codea{bdgReader1}$};
\node[obj] (bread2) at (306:4cm) {$\codea{bdgReader2}$};

\node[obj] (faman) at (342:4cm) {$\codea{fireMngr}$};

\node[topic] (abut2-pump2) at ($1/2*(psens1) + 1/2*(light2) - (.9,.8) + (-.5,.25)$) {};
\node[topic] (psens2-light1) at ($1/2*(psens2) + 1/2*(light1) + (.5,1) + (-.5,.25)$) {};
\node[topic] (psens2-light2) at ($1/2*(psens2) + 1/2*(light2) + (-.5,.25)$) {};
\node[topic] (psens1-log) at ($1/2*(psens1) + 1/2*(log) - (1,-0.25) + (-.5,.25)$) {};
\node[topic] (psens2-log) at ($1/2*(psens2) + 1/2*(log) + (-0.5,0.5) + (-.5,.25)$) {};
\node[topic] (lock1-log) at ($1/2*(lock1) + 1/2*(log) - (0,0) + (1.5,-.25)$) {};
\node[topic] (lock2-log) at ($1/2*(lock2) + 1/2*(log) - (0,0) + (1.5,.25)$) {};
\node[topic] (lock2-psens2) at ($1/2*(lock2) + 1/2*(psens2) + (-.5,.25)$) {};
\node[topic] (bread1-AClist) at ($1/2*(bread1) + 1/2*(AClist) + (0,2) + (-.5,.25)$) {};
\node[topic] (bread2-AClist) at ($1/2*(bread2) + 1/2*(AClist) + (-.5,.25)$) {};
\node[topic] (AClist-lock1) at ($1/2*(AClist) + 1/2*(lock1) + (1.5,-1.4) + (0,.25)$) {};
\node[topic] (AClist-lock2) at ($1/2*(AClist) + 1/2*(lock2) + (1,-.25)$) {};
\node[topic] (faman-lock1) at ($1/2*(faman) + 1/2*(lock1) + (-.5,.25)$) {};
\node[topic] (faman-lock2) at ($1/2*(faman) + 1/2*(lock2) + (-.5,.25)$) {};
\node[topic] (faman-pump1) at ($1/2*(faman) + 1/2*(pump1) - (0,0.5) + (-.5,.25)$) {};
\node[topic] (faman-pump2) at ($1/2*(faman) + 1/2*(pump2) - (0,0) + (-.5,.25)$) {};
\node[topic] (faman-aalarm1) at ($1/2*(faman) + 1/2*(aalarm1) - (0.5,-0.25) + (-.5,.25)$) {};
\node[topic] (faman-aalarm2) at ($1/2*(faman) + 1/2*(aalarm2) - (.5,1.2) + (-.5,.25)$) {};
\node[topic] (abut1-lock1) at ($1/2*(abut1) + 1/2*(lock1) + (-.5,.25)$) {};
\node[topic] (abut1-lock2) at ($1/2*(abut1) + 1/2*(lock2) - (1.75,0) + (-.5,.25)$) {};
\node[topic] (abut1-pump1) at ($1/2*(abut1) + 1/2*(pump1) - (1,2) + (-.5,.25)$) {};
\node[topic] (abut1-pump2) at ($1/2*(abut1) + 1/2*(pump2) - (1,0) + (-.5,.25)$) {};
\node[topic] (abut1-aalarm1) at ($1/2*(abut1) + 1/2*(aalarm1) + (-.25,1.5) + (-.5,.25)$) {};
\node[topic] (abut1-aalarm2) at ($1/2*(abut1) + 1/2*(aalarm2) + (-.5,.25)$) {};
\node[topic] (abut2-lock1) at ($1/2*(abut2) + 1/2*(lock1) + (-.5,.25)$) {};
\node[topic] (abut2-lock2) at ($1/2*(abut2) + 1/2*(lock2) + (-1,1) + (-.5,.25)$) {};
\node[topic] (abut2-pump1) at ($1/2*(abut2) + 1/2*(pump1) - (0,2.5) + (-.5,.25)$) {};
\node[topic] (abut2-aalarm1) at ($1/2*(abut2) + 1/2*(aalarm1) - (1,-0.5) + (-.5,.25)$) {};
\node[topic] (abut2-aalarm2) at ($1/2*(abut2) + 1/2*(aalarm2) - (.5,-.5) + (-.5,.25)$) {};

\draw[arrow] (psens2) -- (psens2-light1);
\draw[arrow] (psens2-light1) -- (light1);

\draw[arrow] (psens2) -- (psens2-light2);
\draw[arrow] (psens2-light2) -- (light2);

\draw[arrow] (psens1) -- (psens1-log);
\draw[arrow] (psens1-log) -- (log);
\draw[arrow] (psens2) -- (psens2-log);
\draw[arrow] (psens2-log) -- (log);
\draw[arrow] (lock1) -- (lock1-log);
\draw[arrow] (lock1-log) -- (log);
\draw[arrow] (lock2) -- (lock2-log);
\draw[arrow] (lock2-log) -- (log);
\draw[arrow] (lock2) -- (lock2-psens2);
\draw[arrow] (lock2-psens2) -- (psens2);
\draw[arrow] (bread1) -- (bread1-AClist);
\draw[arrow] (bread1-AClist) -- (AClist);
\draw[arrow] (bread2) -- (bread2-AClist);
\draw[arrow] (bread2-AClist) -- (AClist);
\draw[arrow] (AClist) -- (AClist-lock1);
\draw[arrow] (AClist-lock1) -- (lock1);
\draw[arrow] (AClist) -- (AClist-lock2);
\draw[arrow] (AClist-lock2) -- (lock2);
\draw[arrow] (faman) -- (faman-lock1);
\draw[arrow] (faman-lock1) -- (lock1);
\draw[arrow] (faman) -- (faman-lock2);
\draw[arrow] (faman-lock2) -- (lock2);
\draw[arrow] (faman) -- (faman-pump1);
\draw[arrow] (faman-pump1) -- (pump1);
\draw[arrow] (faman) -- (faman-pump2);
\draw[arrow] (faman-pump2) -- (pump2);
\draw[arrow] (faman) -- (faman-aalarm1);
\draw[arrow] (faman-aalarm1) -- (aalarm1);
\draw[arrow] (faman) -- (faman-aalarm2);
\draw[arrow] (faman-aalarm2) -- (aalarm2);
\draw[arrow] (abut1) -- (abut1-lock1);
\draw[arrow] (abut1-lock1) -- (lock1);
\draw[arrow] (abut1) -- (abut1-lock2);
\draw[arrow] (abut1-lock2) -- (lock2);
\draw[arrow] (abut1) -- (abut1-pump1);
\draw[arrow] (abut1-pump1) -- (pump1);
\draw[arrow] (abut1) -- (abut1-pump2);
\draw[arrow] (abut1-pump2) -- (pump2);
\draw[arrow] (abut1) -- (abut1-aalarm1);
\draw[arrow] (abut1-aalarm1) -- (aalarm1);
\draw[arrow] (abut1) -- (abut1-aalarm2);
\draw[arrow] (abut1-aalarm2) -- (aalarm2);

\draw[arrow] (abut2) -- (abut2-lock1);
\draw[arrow] (abut2-lock1) -- (lock1);
\draw[arrow] (abut2) -- (abut2-lock2);
\draw[arrow] (abut2-lock2) -- (lock2);
\draw[arrow] (abut2) -- (abut2-pump1);
\draw[arrow] (abut2-pump1) -- (pump1);
\draw[arrow] (abut2) -- (abut2-pump2);
\draw[arrow] (abut2-pump2) -- (pump2);
\draw[arrow] (abut2) -- (abut2-aalarm1);
\draw[arrow] (abut2-aalarm1) -- (aalarm1);
\draw[arrow] (abut2) -- (abut2-aalarm2);
\draw[arrow] (abut2-aalarm2) -- (aalarm2);

\node[topic,teal] (faman-elev) at ($1/2*(faman) + 1/2*(elev) + (-.25,.25) + (-.5,.25)$) {};
\node[topic,teal] (abut2-elev) at ($1/2*(abut1) + 1/2*(elev) - (1.25,-.75) + (-.5,.25)$) {};
\node[topic,teal] (abut1-elev) at ($1/2*(abut2) + 1/2*(elev) - (.5,0) + (-.5,.25)$) {};
\node[topic,teal] (ssens1-faman) at ($1/2*(ssens1) + 1/2*(faman) + (.2,0) + (-.5,.25)$) {};
\node[topic,teal] (ssens2-faman) at ($1/2*(ssens2) + 1/2*(faman) + (-.5,.25)$) {};
\node[topic] (psens1-light1) at ($1/2*(1,1) + 1/4*(psens1) + 1/4*(light1) + (-.5,.25)$) {};
\node[topic,red] (psens1-light2) at ($1/2*(abut2) + 1/2*(pump2) - (1,1.5) + (-.5,.25)$) {};
\node[topic,red] (lock1-psens1) at ($1/2*(lock1) + 1/2*(psens1) + (1, -.5) + (-.5,.25)$) {};

\draw[arrow,thick] (psens1) -- (psens1-light1);
\draw[arrow,thick] (psens1-light1) -- (light1);

\draw[arrow,red,thick] (psens1) -- (psens1-light2);
\draw[arrow,red,thick] (psens1-light2) -- (light2);

\draw[arrow,teal,thick] (abut2) -- (abut2-elev);
\draw[arrow,teal,thick] (abut2-elev) -- (elev);
\draw[arrow,teal,thick] (ssens1) -- (ssens1-faman);
\draw[arrow,teal,thick] (ssens1-faman) -- (faman);
\draw[arrow,teal,thick] (ssens2) -- (ssens2-faman);
\draw[arrow,teal,thick] (ssens2-faman) -- (faman);

\draw[arrow,teal,thick] (abut1) -- (abut1-elev);
\draw[arrow,teal,thick] (abut1-elev) -- (elev);
\draw[arrow,teal,thick] (faman) -- (faman-elev);
\draw[arrow,teal,thick] (faman-elev) -- (elev);

\draw[arrow,red,thick] (lock1) -- (lock1-psens1);
\draw[arrow,red,thick] (lock1-psens1) -- (psens1);

\end{tikzpicture}}
\caption{Information Flow Graph for the Building Automation System.}
\label{fig:sifg}
\end{figure}
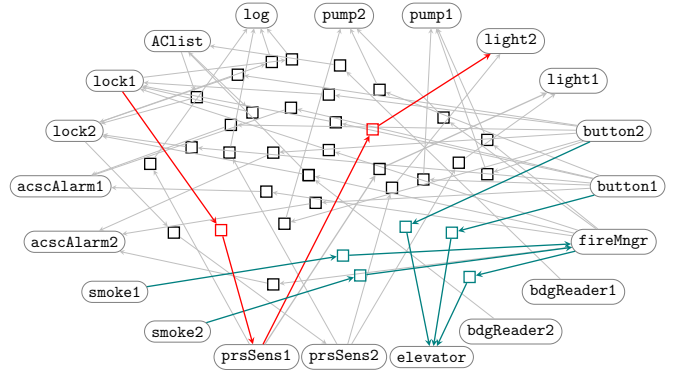

\begin{example}\label{ex:violating:sym:path}
\Cref{fig:sifg} contains the symbolic information flow of our BAS running example.
The red path is 
${\code{lock1}} \rightarrow F_{c_{\codea{lock1}},c_{\codea{prsSens1}}} \rightarrow {\code{prsSens1}} \rightarrow F_{c_{\codea{prsSens1}},c_{\codea{light2}}} \rightarrow {\code{light2}}$.
The violating path of \Cref{ex:violating:path} is one of its instances where the topics $t_1$ and $t_2$ satisfy the formula $F_{c_{\codea{lock1}},c_{\codea{prsSens1}}}$ and $F_{c_{\codea{prsSens1}},c_{\codea{light2}}}$. 
\end{example}

The following theorem establishes that the symbolic information flow graph provides a sound and complete representation of the information flow graph which in turn is proved to match the information flow permitted by the configuration according to the operational semantics of the MQTT broker given in~\Cref{sec:broker:sem}.
\begin{restatable}{theorem}{correct}
\label{th:correctness}
	Given a configuration,
    there is an information flow from $d$ to $d'$ 
    if and only if there is a symbolic information flow from $d$ to $d'$.
\end{restatable}

The next section presents the implementation of our tool for checking information flow in IoT systems.
However, we highlight a crucial point: 
the symbolic information flow graph cannot be computed solely through operations on regular languages
(as one might initially assume, given that both AWS and MQTT wildcards can be described using regular language operators).
The issue arises from the substitution of the variable $\cidVAR$ within the policy rules that complicates the task.
In particular, $\Sin_c$ and $\Sout_c$ are not always regular languages.
For example, consider a certificate $c$ with 
\[
\Phi(c) = \{ (\allowa, \iotConn, \ast),  (\allowa, \iotPub, \cidVAR/\cidVAR) \}.
\]
The output set of topics over which $c$ can publish $\Sout_c = \{ \omega/\omega \mid \omega \in (\Lambda \cup \{ \code{/} \})^{*} \} $ is \emph{not} a regular language, as it contradicts the pumping lemma~\cite{pumping}.

\begin{figure}[t]
\centering 

\scalebox{.85}[.85]{
\begin{tikzpicture}[
	node distance=8mm,
	arrow/.style={-stealth,thick},
	iarrow/.style={-Straight Barb,double},
	obj/.style={font=\footnotesize},
	weight/.style={font=\normalsize,midway},
	title/.style={font=\small\color{black!50}\ttfamily},
    typetag/.style={rectangle, draw=black!50, font=\small\ttfamily, anchor=west},
    topic/.style={rectangle, thick, draw=black, font=\small\ttfamily}]

\node (zero) at (0,0) {
	\begin{tabular}{c}
	\includegraphics[width=5mm,keepaspectratio]{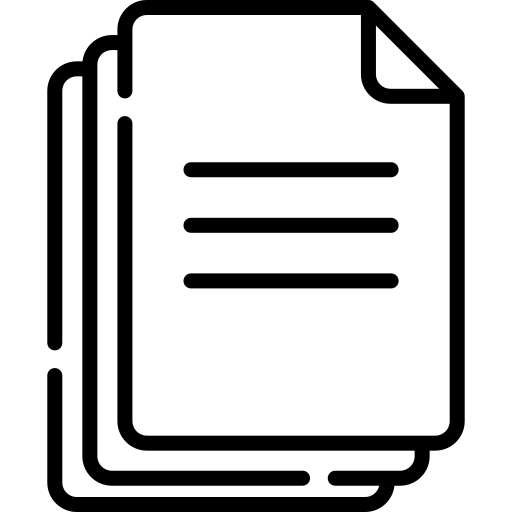}\\
	certificates\\
	and policies
	\end{tabular}
};

\node [obj, right=15mm of zero] (ppar) {
	\begin{tabular}{c}
	\includegraphics[width=5mm,keepaspectratio]{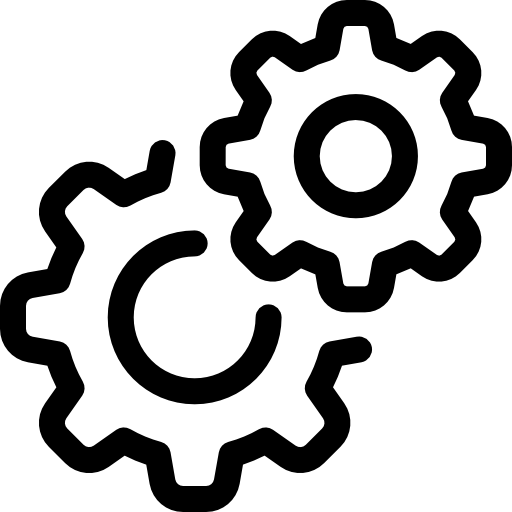}\\
	policy\\
	parser
	\end{tabular}
};
\node [obj, right=3mm of ppar] (z3) {
	\begin{tabular}{c}
	\includegraphics[width=5mm,keepaspectratio]{img/gear}\\
	SAT\\
	solver
	\end{tabular}
};

\node [obj, below=5mm of z3] (netx) {
	\begin{tabular}{c}
	\includegraphics[width=5mm,keepaspectratio]{img/gear}\\
	Network\\
	X
	\end{tabular}
};
\node [obj, left=3mm of netx] (qpar) {
	\begin{tabular}{c}
	\includegraphics[width=5mm,keepaspectratio]{img/gear}\\
	query\\
	parser
	\end{tabular}
};

\node [draw=black!40, obj, right=15mm of $1/2*(z3) + 1/2*(netx)$] (sifg) {
	\begin{tabular}{c}
	\begin{tikzpicture}
	\node[circle, draw, inner sep=0.7mm] (n1) at (0,0.2) {};
	\node[circle, draw, inner sep=0.7mm] (n2) at (0,-.1) {};
	\node[circle, draw, inner sep=0.7mm] (n1') at (.5,.3) {};
	\node[circle, draw, inner sep=0.7mm] (n2') at (.9,0) {};
	\node[circle, draw, inner sep=0.7mm] (n3') at (.5,-.3) {};
	\node[rectangle, thick, draw, inner sep=0.7mm] (t) at (0.5,0) {};
	\draw [arrow] (n1) -- (t);
	\draw [arrow] (n2) -- (t);
	\draw [arrow] (t) -- (n1');
	\draw [arrow] (t) -- (n2');
	\draw [arrow] (t) -- (n3');
	\end{tikzpicture}\\
	symbolic\\
	information\\
	flow graph\\
	\end{tabular}
	};

\node [obj, left=19mm of qpar] (usr) {
	\begin{tabular}{c}
	\includegraphics[width=7mm,keepaspectratio]{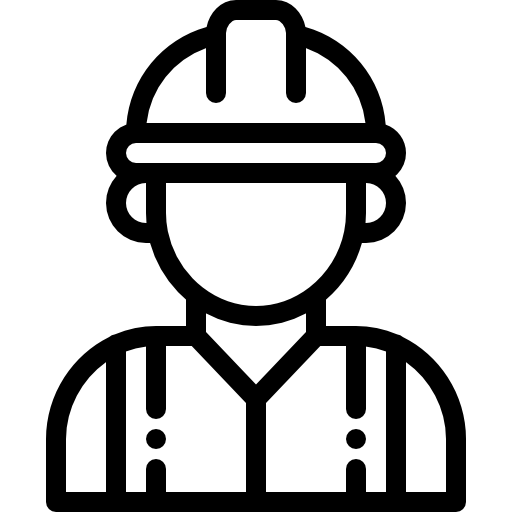}\\
	\normalsize user
	\end{tabular}
};

\node [draw=black!40, inner ysep=0mm, inner xsep=0mm, fit={(ppar) (z3)}] (siflgen) {};
\node [title, below right=.3mm and -10mm of ppar] (siflgent) {graph builder};
\node [draw=black!40, inner ysep=0mm, inner xsep=0mm, fit={(qpar) (netx)}] (qman) {};
\node [title, below right=.3mm and -10mm of qpar] (qmant) {query manager};
\begin{scope}[on background layer]
\node [fill=black!8, inner ysep=1mm, inner xsep=2.5mm, fit={(ppar) (qman) (qmant) (sifg) (qpar) (siflgen) (siflgent)}] (tool) {};
\end{scope}

\draw [arrow] (zero) -- (siflgen);
\draw [iarrow] (siflgen) -- (sifg);
\draw [iarrow] (qman) -- (sifg);
\draw [iarrow] (ppar) -- (z3);
\draw [iarrow] (qpar) -- (netx);
\draw [stealth-,thick] (usr)
	edge[bend right] node[above,xshift=-2mm] {
    \footnotesize 
    \begin{tabular}{c}
	reply and\\
	witness
	\end{tabular}} ($(qman.south west)!0.5!(qman.west)$);
\draw [-stealth,thick] 
	(usr)
	edge[bend left] node[below,xshift=-2mm,yshift=0mm] {
    \footnotesize query} 
    ($(qman.north west)!0.5!(qman.west)$);

\end{tikzpicture}
}
\caption{Overview of~\thetool.}
\label{fig:overview}
\end{figure}

\section{Implementation and Evaluation} \label{sec:tool}
We implemented our verification approach 
in \thetool~\cite{repo}, an open-source AWS IoT-policy analyzer written in Python.
Its architecture is shown in \Cref{fig:overview}: the tool takes a configuration as input, parses the policies using the library Policy Universe~\cite{policyuniverse}, and uses the SMT solver cvc5~\cite{cvc5} to construct the symbolic information-flow graph.
Finally, a query manager answers users' requests
by checking reachability in the graph via the NetworkX library~\cite{netx}.

The following subsections first describe how the symbolic information flow graph is built; then clarify which queries can be performed on the resulting graph; and finally, evaluate the effectiveness and scalability of \thetool.

\subsection{Building the graph}
\begin{algorithm}[t]
	\caption{Information Flow Graph Construction.}\label{al:sifg}
	\begin{algorithmic}[1]
		\Require a configuration $(\certs, \pols, \varphi, \Dev, k)$
		\Ensure its symbolic information flow graph
		
            \State $(N, E) \gets (\Dev, \emptyset)$
	
		\ForAll{$(c, c') \in \certs \times \certs$}
		\State build $F_{c,c'}$ as in~\Cref{def:info-graph:sym}
        \State check its satisfiability through~\Cref{al:witness}
		\If{$F_{c,c'} \text{ is satisfiable}$}\label{al:sifg:checks}
		\State $N \gets N \cup \{F_{c,c'}\}$
		\ForAll{$d$ such that $c \in k(d)$}
		\State $E \gets E \cup \{(d, F_{c,c'})\}$
		\EndFor
		\ForAll{$d'$ such that $c' \in k(d')$}
		\State $E \gets E \cup \{(F_{c,c'},d')\}$
		\EndFor
		\EndIf
		\EndFor
		\State \Return $(N, E)$
	\end{algorithmic}
\end{algorithm}

\begin{algorithm}[t]
    \caption{Witness for the predicate $F_{c,c'}$}\label{al:witness}
    \begin{algorithmic}[1]
		\Require two certificates $c$, $c'$
		\Ensure a witness $w = (\cid, \cid', t, \tf)$ or None

        \State $S \gets \reqEval{(c, \cid, \iotConn, \cid)} 
            \land \reqEval{(c', \cid', \iotConn, \cid')}
            \land \reqEval{(c, \cid, \iotPub, t)}
       $\State $\qquad\qquad
			\land\ \reqEval{(c', \cid', \iotRec, t)}
            \land \reqEval{(c', \cid', \iotSub, \tf)}$ 
            \label{al:witness:cond}
            
        \If{$\text{SMT}(S)$ \text{is unsat}}
           \State
            \Return None   \Comment{Early Fail} \label{al:witness:exit}
        \EndIf
    
        \State $S_e \gets S \land (t = \tf)$  
        \label{al:witness:easy}
        \If{$\text{SMT}(S_e)$ \text{is sat}}
           \State 
            \Return $w \gets (\cid, \cid', t, \tf)$  \Comment{Early Success}
        \EndIf

        \For{$n = 1, \dots, 8$}   
            \State $S_h \gets S \land t = t_1 / t_2 / \dots / t_{n}$ \label{al:witness:tlv}
            {\State \mbox{$S_h' \gets \tf \in \lang{(t_1 | +) / \dots / (t_{n} | + )}$ \label{al:witness:case1}}}

            \For{$m = 1, \dots, n-1$}
                \State \mbox{$S_h' \gets S_h' \lor  \tf \in \lang{(t_1 | +) / \dots / (t_{m} | +) / \#}$\label{al:witness:case2}}
            \EndFor

            \If{$\text{SMT}(S_h \land S_h')$ \text{is sat}}
               \State 
                \Return $w \gets (\cid, \cid', t, \tf)$
            \EndIf
        \EndFor
        \State \Return None

    \end{algorithmic}
\end{algorithm}

\Cref{al:sifg} describes the procedure to build the graph.
Given a configuration $(\certs, \pols, \varphi, \Dev, k)$, 
we start from a symbolic graph that has a node for each device $d \in \Dev$ and no arcs.
Then, we iterate over each pair of certificates $c$ and $c'$ of $\certs$, building the predicate $F_{c,c'}$ of \Cref{def:info-graph:sym}, encoding it into an SMT formula (via Skolemization), and invoking the solver:
if the formula is satisfiable, 
then we add a node for $F_{c,c'}$, and we connect it with all the devices associated with $c$ (incoming arc) and $c'$ (outgoing arc).

The core challenge is efficiently deciding the SMT problem, which 
amounts to finding two client ids ($\cid, \cid'$), a topic ($t$) and a topic filter ($\tf$)
that satisfy the formula $\Psi(\cid, \cid', t, \tf)$ below (which encodes $\Fout_{c}$ and $\Fin_{c'}$), 
or proving it unsatisfiable.
\begin{gather*}
\Psi(\cid, \cid', t, \tf) = \reqEval{(c, \cid, \iotConn, \cid)} 
            \land \reqEval{(c', \cid', \iotConn, \cid')}\ \land \\
            \reqEval{(c, \cid, \iotPub, t)}
			\land \reqEval{(c', \cid', \iotRec, t)}
            \land \reqEval{(c', \cid', \iotSub, \tf)}
            \land t \in \lang{\mathit{tf}^{\mqtt}}
\end{gather*}
The predicates can be encoded by resorting to the strings and regular expressions theories, but the presence of two different kinds of wildcards (AWS and MQTT) in the policies makes the problem challenging in practice:
the solver 
would take hours and hours of computation to check the satisfiability of a single predicate $F_{c,c'}$, if a naive encoding of the formula is used.
For solving this problem, we implement an effective strategy to invoke the solver in \Cref{al:witness}, which is divided into three stages.
The first two stages are heuristics for early failure and success, while the last one is a complete decision procedure, optimised for dealing with the potency of  wildcards by taking advantage of the peculiar structure of the $F_{c,c'}$ formulas.

Our heuristics are based on checking weakened and strengthened versions of $\Psi$, where the constraint $t \in \lang{\mathit{tf}^{\mqtt}}$ is approximated by simpler conditions.
In the first stage, the algorithm considers a set of conditions that are necessary but not sufficient for guaranteeing communication between $c$ and $c'$.
The constraint $t \in \lang{\mathit{tf}^{\mqtt}}$ is removed and the solver is  asked to simply determine if there exists some assignment that satisfies at least the $\reqEval{\_}$ predicates of $\Psi$.
If this is not the case, then communication is trivially impossible, and an \emph{early fail} occurs (line 4).
Note that these conditions are not sufficient for communication because we are imposing no relation between the topic filter $\tf$ and the topic $t$.
The second stage checks a condition that is sufficient but not necessary for communication: 
whether $c'$ can subscribe to the same topic $t$ on which $c$ can publish (line 7).
Basically, we ignore the MQTT wildcards and run the solver on a version of $\Psi$ where $t = \tf$ substitutes $t \in \lang{\mathit{tf}^{\mqtt}}$.
Although simple, these two heuristics greatly improve our tool's performance, solving most real-world cases.

If neither an early failure nor a solution is found, we adopt the complete resolution strategy (\emph{hard strategy}).
This third stage considers $\Psi$ as it is, also taking MQTT wildcard characters into account.
Finding an efficient SMT encoding of $t \in \lang{\tf^{\mqtt}}$ is made particularly challenging by the fact that $\tf$ is not a constant value, but it is rather dynamically computed as a solution of the other constraints.
To address this problem, we focus on the peculiar format of IoT resources, exploiting the hierarchical structure of topics $t = t_1/t_2/\dots/t_n$ and topic filters $\tf = \tf_1/\tf_2/\dots/\tf_m$, which have a maximum depth of 8 in AWS.
The algorithm essentially creates a disjunction of constraints encoding alternative formats for $t$ and $\tf$, and solves it with a specific search strategy.
In more detail, we start from the assertions of line 2 about the $\reqEval{\_}$ predicates of $\Psi$, and we add the condition of line 9 stating that $t$ has $n$ levels.
Finally, we impose $t \in \lang{\tf^{\mqtt}}$ by requiring $t$ and $\tf$ to satisfy one of the following: 
\begin{itemize}[noitemsep]
	\item the number of levels of $t$ and $\tf$ coincides (i.e., $n = m$) and, $\tf_i$ matches $t_i$ for each level $i$, i.e., $\tf_i = t_i$ or $\tf_i = \code{+}$ (line 10, recall that $\code{+}$ matches any single level of the topic);
	\item $\tf$ has no more levels than $t$ (i.e., $m \leq n$), $\tf_i$ matches $t_i$ for all $i < m$, and the last level of the topic filter  is $\tf_m = \code{\#}$ (line 12, recall that $\code{\#}$ matches any remaining sequence of levels $t_m/t_{m+1}/\dots/t_n$ of the topic  $t$).
\end{itemize}
For example, the fist case tells us that the topic $t = \code{phAC/floor1/prsSens1/enable}$ is included in $\lang{\tf_1^{\mqtt}}$ with $\tf_1 = \code{phAC/{+}/prsSens1/+}$, because each layer of $\tf_1$ is either equal to the corresponding one of $t$ or it is the wildcard character $\code{+}$, which matches both $\code{floor1}$ and $\code{enable}$ (as well as any other alphanumerical string).
Moreover, the second case tells us that the topic above is also included in $\lang{\tf_2^{\mqtt}}$ with $\tf_2 = \code{phAC/\#}$, because the wildcard character $\code{\#}$ matches $\code{floor1/prsSens1/enable}$ (as well as any other string 
over the alphabet of alphanumerical characters enriched with the forward-slash character `$/$').

Our optimisations allow \Cref{al:witness} to effectively check the satisfiability of the formula $F_{c,c'}$, as shown by our experiments below.
An upper bound to the complexity of computing the symbolic graph is given by the following theorem: 
\begin{restatable}[Complexity]{theorem}{complexity}\label{thm:complexity}
Let $(\certs, \pols, \varphi, \Dev, k)$ be a configuration, and let $f(n)$ be the cost for deciding the satisfiability of quantifier-free predicates
of size $n$ using regular expressions and string theories.
Then, the time required by \Cref{al:sifg} is at most 
$O(|\certs|^2 \cdot(|\Dev| + |\overline{\pol}| \cdot |\overline{\RE}| + f(|\overline{\pol}|\cdot |\overline{\RE}|)))$, where $|\overline{\pol}|$ and $|\overline{\RE}|$ are the size of the largest policy and resource expression.
\end{restatable}

Note that $f$ is at least exponential, independently of the used theories, as the problem subsumes SAT, for which only exponential solutions are known.
However, the exponential cost is only on the size of policies and on the number of resource expressions in statements, while the algorithm 
is polynomial on the size of the network (our experimental results are consistent with this estimation).
Noticeably, the size of the network is expected to increase more rapidly than the complexity of the individual policies in practice (as far as policies are adequately designed).

\subsection{Query Manager}\label{app:poker}

\thetool\ allows users to define security labels, i.e.~names $\code{s}, \code{s'}, \code{s''}, \dots$ for sets of devices that play similar roles in the IoT system.
After that, 
\thetool\ can check four kinds of queries on IoT systems. The first verifies the existence of a permitted information flow, while the others correspond to the security properties
of integrity, confidentiality, and isolation.
The query $\code{flow(s,s')}$ checks whether there is a permitted information flow from a device labelled with $\code{s}$ to a device labelled with $\code{s'}$.
The query $\code{onlyAffects(s, s')}$ checks integrity, i.e.~it requires that only a device labelled with $\code{s}$ can influence the behaviour of a critical device labelled with $\code{s'}$. 
The query $\code{onlyKnows(s, s')}$ checks confidentiality, i.e.~it requires that only devices labelled $\code{s}$ can receive private information produced by devices labelled $\code{s'}$.
Finally, the query $\code{isolated(s, s'})$ requires that $s$-labelled devices and $s'$-labelled devices cannot interact.

By~\Cref{th:correctness}, query evaluation reduces to
reachability in the symbolic graph: each query can be solved by standard search algorithms, with worst-case performance linearly proportional to the number of edges of the graph.
More specifically, $\code{flow(s,s')}$ (and $\code{onlyKnows(s', s)}$) are evaluated by visiting the graph from the nodes associated with $s$-labelled devices, and checking if at least one of (respectively, all) the visited nodes corresponds to an $s'$-labelled device. 
The same approach applies to $\code{onlyAffects(s, s')}$, but with the edges reversed.
Finally, $\code{isolated(s, s')}$ requires two traversals of the graph for checking that $s$ cannot reach $s'$ and vice-versa.

\subsection{Effectiveness Evaluation}

\begin{table}[t]
\caption{Analysing the BAS with \thetool.}
\label{tab:query}
\centering
{
\footnotesize
\setlength{\tabcolsep}{3pt}
\begin{tabular}{l c c }
    \toprule
    \textbf{Query} & \textbf{Result} & \textbf{Counterexample} \\
    \midrule
    $\codea{isolated(phyAC1, elevator)}$ & \cmark & \\
    $\codea{onlyAffects(fire, elevator)}$ & \cmark & \\
    $\codea{onlyKnows(phyAC1, lock1)}$ & \xmark & $\codea{lock1; prsSens1; light2}$ \\
    \bottomrule
\end{tabular}
}
\end{table}

We first evaluate the effectiveness of \thetool\ using the case study of \Cref{sec:case:study}. 
Due to the lack of open-source configurations for multi-device IoT systems, we constructed a representative BAS that captures the structural rules, constraints, and communication patterns of typical IoT scenarios.

We assume that each device is configured with a known IoT policy. The complete configuration is available in the online repository~\cite{repo} and includes 17 devices, 20 certificates, and 15 IoT policies, each of which consists of approximately 20 lines of code. Our threat model assumes that the cloud infrastructure is trusted, whereas devices may be vulnerable and IoT policies may be misconfigured. 
An attacker may gain full control of a vulnerable device, such as a light bulb, and maliciously select the topics to which it publishes or subscribes, exploiting overly permissive policies to compromise critical devices or access sensitive information.
Security labels identify the subsystem to which each device or cloud service belongs. For example, $\code{phyAC1}$ labels the devices involved in the physical access control of the first floor, such as $\code{bdgReader1}$, together with the related cloud services $\code{log}$ and $\code{AClist}$. Similarly, $\code{fire}$ labels the fire-alarm manager, the alarm buttons, and the smoke sensors on both floors.

We use \thetool\ to analyse the IoT policies, reconstruct the permitted device interactions, and represent them in the information-flow graph shown in \Cref{fig:sifg}. 
\Cref{tab:query} reports the queries considered in our evaluation and the corresponding results produced by \thetool. 
The first query evaluates an \emph{isolation} property: the elevator must be completely isolated from devices in $\code{phyAC1}$ subsystem. 
\thetool\ verifies that no path violates this requirement: no path starts from the elevation, and the ones reaching it -- highlighted in \textcolor{teal}{\bf teal} in \Cref{fig:sifg} -- do not come from $\code{phyAC1}$.
The second query evaluates the \emph{integrity} of the elevator, whose behaviour should depend only on devices labelled $\code{fire}$. 
The property is satisfied: 
\textcolor{teal}{\bf teal} paths originate from the intended devices.
The third query, $\codea{onlyKnows(phyAC1, lock1)}$, evaluates \emph{confidentiality}: information produced by the first-floor lock should be accessible only to devices and services labelled $\code{phyAC1}$. \thetool\ detects a violation, as a compromised light bulb can receive messages intended for other light bulbs, allowing information originating from $\code{lock1}$ to reach $\code{light2}$ on the second floor. This unintended flow results from the overly permissive policy for the light bulbs, 
which is shown in \Cref{fig:light:bulb:policy}.

\subsection{Scalability Evaluation}

\begin{table}[t]
    \caption{\thetool\ scalability over real-world configurations.}\label{tab:rw:res}
    \centering
    \footnotesize
    \setlength{\tabcolsep}{4.5pt}
        \begin{tabular}{c c | c c c c c c}
		\toprule
		& & \textbf{} & \textbf{hard}  & \multicolumn{3}{c}{\textbf{building}} & \textbf{query} \\
        & & \textbf{nodes} & \textbf{strategies}  & \multicolumn{3}{c}{\textbf{time ($s$)}} & \textbf{time ($s$)} \\
        &  & avg.   &  avg.   & min. & avg. & max. & avg. \\
		\midrule
		\multirow{ 13}{*}{\rotatebox{90}{\footnotesize \textbf{certificates}}} & 20 	&   264.6 	&  1.8  &  0.42   &  2.87  & 12.00  & 0.0060 \\
		& 40 	&   973.0	&  4.5 	&  2.30   &  6.28  & 16.55  & 0.0086 \\
		& 60 	&  2211.2	&  5.6	&  2.93   &  9.22  & 23.39  & 0.0108 \\
		& 80 	&  3903.9	& 13.6	&  5.93   & 16.34  & 29.08  & 0.0136	\\
		& 100 &  6098.1	& 16.7	&  7.50   & 19.46  & 36.55  & 0.0167	\\
		& 120 &  8764.2	& 21.5	& 13.58   & 24.69  & 41.49  & 0.0198	\\
		& 140 & 12080.1	& 28.7	& 18.61   & 31.06  & 46.63  & 0.0244	\\
		& 160 & 15415.9	& 29.6	& 20.03   & 33.01  & 44.43  & 0.0280	\\
		& 180 & 19571.7	& 37.0	& 26.61   & 37.87  & 52.79  & 0.0318	\\
		& 200 & 24209.9   & 45.8	& 31.62   & 46.06  & 57.34  & 0.0349	\\
		& 220 & 29066.2   & 50.5	& 32.39   & 50.93  & 62.02  & 0.0393	\\
		& 240 & 34732.6   & 59.0	& 42.11   & 57.68  & 62.41  & 0.0430	\\
		& 258 & 40019.0   & 65.0	& 62.35   & 62.67  & 63.53  & 0.0473	\\
		\bottomrule
\end{tabular}

\end{table}

We experimentally evaluate the scalability of \thetool\ on real-world configurations when the IoT configuration's size (number of certificates) grows, using cvc5 as our SMT solver.
We conduct all our experiments below on a desktop machine with an i7-10700K processor (3.80GHz) and 32GB RAM, running Windows 11.
We consider 258 real-world policies originally implemented for various IoT devices by different manufacturers, available online in the benchmark of the P-verifier tool~\cite{p-verifier-website}.
Using these policies, we performed 30 experiments as follows:
\begin{enumerate*}[label=(\roman*)]
    \item we generate 13 configurations with an increasing number of devices from 20 to 258; 
    \item each device is randomly associated with a distinct certificate, real-world policy (from the benchmark), and a fresh security label;
    \item we build the symbolic information flow graph of each configuration, measuring its computation time;
    \item we interrogate \thetool\ with 1000 reachability queries 
    between two random devices, and record the response time.
\end{enumerate*}

The results are in~\Cref{tab:rw:res}, where the first column indicates the size of the IoT configuration (i.e., the number of certificates), the second reports the number of nodes of the graph (both devices and symbolic nodes representing sets of topics), the third one is the number of times the hard strategy was used to determine if two devices communicate, and the last two columns contain the time needed for building the graph and solving the queries.
For each configuration size, the table reports the arithmetic mean of the results obtained by the configurations of that size, across the 30 experiments.
Moreover, for the graph build time, which is the most expensive part of the computation, we also mention the fastest and slowest runs.
When considering all the 258 real-world policies of the benchmark, \thetool\ builds the symbolic graph in less than 64 seconds. 
Moreover, the time required to build the graph increases linearly 
with the number of times the algorithm resorts to the hard strategy.
In more detail, the time needed to solve the SMT problem with the hard strategy is 0.2589 seconds on average (calculated across all 258 policies), and the worst case is 8.7353 seconds;  early fail and success take 0.0220 and 0.0122 seconds on average instead, with worst cases 1.2852 and 0.0905 seconds, respectively.
While running~\Cref{al:witness} on all 258 policies, early failure occurs 90.4\% of the time, early success 7.1\%, and hard strategy 2.5\%.
The time needed for answering each slot of 1000 queries is less than $0.1$ seconds on average.
In conclusion, the algorithm scales well up to configurations with more than 250 devices: the one-time effort for building the graph is reasonable, and evaluating queries is almost immediate.

A version of~\thetool\ also exists, which employs Z3~\cite{z3} as SMT solver in place of cvc5.
We repeated our experiments with Z3, using the same configurations generated for cvc5.
The detailed experimental results are in Appendix~\ref{app:z3}, from which it emerges that cvc5 performs significantly better than Z3.

\section{Discussion}\label{sec:ext:model}%

The model presented in \Cref{sec:model} focuses on the core communication mechanisms of MQTT and the main components of AWS IoT Core policies. While this modeling choice keeps the formal treatment tractable, it omits certain protocol and platform features available in practice that we do not capture. For example, we do not consider wildcard matching in ARN components other than the resource field, such as the region, or the possibility of temporary credentials.
In this section, we discuss how the formalization can be extended to capture and account for some of these features when verifying information-flow properties.

Among the advanced features of AWS IoT policies that we omitted, \emph{things} and \emph{conditions} are those that we can incorporate more naturally by extending our model.
Policy developers can define named virtual devices, called \emph{things}, in their configurations, which are associated with a set of custom attributes, e.g., geographical location, source IP address or domain name. 
Formally, a thing can be represented as a set of bindings from attribute names to strings which may be used in resource expressions and substituted with their values when the policy is instantiated, similarly to what is done for $\cidVAR$.
For example, in our BAS scenario of \Cref{sec:case:study} one may define the attributes $\code{\iotvar{floor}}$ and $\code{\iotvar{room}}$, and add the resource expression $\code{floornum\iotvar{floor}/roomnum\iotvar{room}/light?}$ in the policy.
As a result, two things that associate different values with the attributes above may publish on disjoint sets of topics, even if they share the same policy.
Our formal model can be updated by considering an additional special case for things that
are associated with connecting devices during certificate authentication, and a specific certificate is either for things or for common devices; hence, the kind of instantiation to use is always clear.
When the broker receives a request $r$ from a device, the policy instance governing the request is obtained by taking $\Phi(c)$ and applying the substitutions defined by the attributes of the authenticated device.
Updating the definition of the information flow graph and the procedure for building it in~\Cref{sec:infoflow} and \ref{sec:tool} is trivial, since things are inherently more constrained than common devices.
While we kept the formal development as simple as possible, our tool is capable of handling configurations that use things, provided their bindings are specified in advance.

An additional feature of policy statements is the optional ``Condition'' field that further restricts the cases in which the statement is considered when evaluating a request.
Conditioning expressions are built using a predefined finite set of operators, keys, and values~\cite{IAM}.
Such conditions require only a mild extension of our model by enriching the policy statements in \Cref{def:pol} with suitable (decidable) predicates over requests.
Then, \Cref{def:perm,def:ioset,def:eval} are updated by considering these predicates in conjunction with $\res \in \lang{\re^{\AWS}}$.

Finally, we argue that focusing on AWS IoT Core as a specific instance of IoT-Cloud deployments does not limit the generality of our methodology.
On the one hand, focusing on a specific case enables us to concretely apply our verification approach by targeting real-world configurations.
On the other hand, several MQTT brokers support authorization policies for publishing messages and subscribing to topics, e.g., Azure IoT Hub~\cite{IoTAzureAC} and HiveMQ~\cite{HiveMQAC}.
These policies are typically written in an RBAC policy language that supports wildcards to specify sets of topics.
We are confident that our methodology can be easily adapted to these MQTT brokers with very few modifications in the technical aspects, e.g., encoding roles 
through suitable predicates or slightly adjusting the treatment of wildcards. 

\section{Related Work}\label{sec:ralated}

Several papers have investigated the problem of specifying and verifying cloud-based access policies and information flow in IoT systems.
However, to the best of our knowledge, no paper in the literature addresses the verification of information flow properties in cloud-based IoT access policies.

Backes et al.~\cite{Zelkova} proposed Zelkova, a tool that statically analyses policies by encoding them into SMT formulas and using off-the-shelf solvers to verify specific requirements.
Zelkova can also compare different policies based on permission inclusion.
Building on that, D’Antoni et al.~\cite{DAntoniProjective2024} proposed an ad hoc SMT encodings for efficiently checking whether a policy is public (i.e. if the number of allowed IPs exceeds a given threshold).
With a similar approach, Eiers et al.~\cite{BultanRepair,BultanPerm,BultanQuacky} use an SMT solver to quantify the permissiveness of policies and to reduce the number of allowed requests below some predefined threshold.
Jin et al.~\cite{p-verifier} proposed P-Verifier, a tool to detect security flaws in policy configurations. 
They focus on  specific flaws like wildcard injection
that may result in the granting of unwanted permissions. 
Their verification mechanism relies on translating the policies into suitable SMT formulas and checking them with an off-the-shelf solver.
Barnett et al.~\cite{BarnettModeling2025} proposed a modular formalization of the AWS authorization engine and a corresponding SMT-based analysis tool, called IAM-MULTIPOLICYANALYZER, for verifying properties pertaining to multiple policies of different types (identity-based, resource-based policies, service and resource control policies, permissions boundaries).
The technical aspect of this work is similar to ours, namely determining the effect of combined policies.
However, they investigate how multiple policies affect permitted operations when applied together, whereas we focus on simpler policies and check whether the operations permitted to different entities allow communication.
All these papers use an SMT solver to verify that a policy satisfies some requirements. 
However, unlike our work, they do not give a formal model of IoT systems, nor do they consider information flow among them.
Since the properties they consider differ from ours, their verification mechanism is also different: they only analyse each policy in isolation to determine the permissions granted.
Consequently, it is not straightforward to apply their methodology and tools to verify information-flow properties within a network of devices, as this requires consideration of several policies.
In contrast, our verification mechanism considers all policies simultaneously and all possible device interactions to construct the information flow graph. 
This requires matching the (possibly infinite) sets of topics over which two devices may publish and subscribe, which is challenging because topics may depend on variables, like the client id, and may use wildcards.

Regarding information security in IoT systems, Bastys et al.~\cite{Bastys18} investigated the problem of securing IoT apps that use the If-This-Then-That paradigm.
Similarly, Yu et al.~\cite{TAPFixer} propose TAPFixer to detect and repair vulnerabilities in the setting of Home Automation.
They perform model-checking on the Trigger-Action Programming (TAP) rules of the analysed system.
Celik et al.~\cite{Celik19} proposed IoTGuard, a dynamic, policy-based enforcement system that monitors the runtime behaviour of IoT apps to detect violations of safety and security policies. 
Mandalari et al.~\cite{IoTrim} propose another dynamic approach to automatically classify and block the non-essential network traffic of IoT devices.
Some works~\cite{BodeiG17,Balliu21} propose process-algebra models that capture the core aspects of the behaviour of IoT systems,
and use formal methods (control-flow analysis, type systems) to prove various security properties.
All these papers share a goal similar to ours: ensuring that information flows within an IoT system complies with a given security policy. 
However, 
their approach is orthogonal to ours:
they assume access to the app's code and
take into account the interaction between sensors and actuators mediated by the environment, but they
ignore access control policies.
In contrast, we treat IoT devices as black boxes, over-approximating their behaviour based on the actions their IoT policies allow:
we consider an attacker capable of taking full control of a compromised device, subverting its code.
As a result, we do not need to re-analyse systems when a device is replaced or its behaviour changes due to an attack, as long as the access policy remains unchanged.
Moreover, to streamline the discussion, we isolate the information flow permitted by the broker's access control policy; we leave to future work the analysis of scenarios involving indirect interactions between sensors and actuators. 
To this end, we anticipate no significant challenges in integrating our information flow graph with arcs obtained with the approaches discussed above.

\section{Conclusion}\label{sec:concl}

We addressed the problem of verifying information flow in IoT access policies.
We focused on AWS IoT Core's main components as a specific instance of IoT-Cloud deployments.
First, we introduced a formal model of IoT policies,
then, we built an information flow graph to capture communication between IoT devices via MQTT topics. 
We reduce the problem of checking information flow between devices to a graph reachability problem. 
To enhance practical feasibility, we introduced a symbolic version of the graph, replacing topics with logic formulas to succinctly represent potentially infinite sets of topics.
Finally, we implemented our verification mechanism in the \thetool\ tool,
we evaluated its effectiveness on a representative scenario and its
scalability on a collection of real-world IoT policies. 

Future work aims to improve our analysis and foster the adoption of \thetool\ by practitioners.
A first extension involves considering the typical dynamicity of IoT systems, where devices may join and leave. 
Then, we will implement a change-impact analysis to determine how a policy update affects the information flow and if it preserves the security requirements.
Finally, we plan to consider scenarios where information flows 
occur through indirect interaction between devices' sensors and actuators.

\subsection*{Acknowledgements}
This work has been partially supported by RDS PTR 25-27 CYBER 2.1 "Progetto Cybersecurity" WP3 LA 3.21 - CUP: D63C24001060001. 

\bibliographystyle{IEEEtran}
\bibliography{biblio}

\clearpage

\appendices
\begin{table}[t]
	\caption{Summary of symbols and notation used in the paper.}
	\label{tab:notation}
	\centering
	\scriptsize
	\setlength{\tabcolsep}{4.5pt}
	\def\arraystretch{1.35}%
	\begin{tabular}{c l}
		\toprule
		\textbf{Notation}    & \textbf{Description}  \\
		\midrule
		$\cids$, $\topics$, $\topicfs$ & the set of client ids, topics, and topic filters\\
		$\codea{\cidVAR}$ & the variable $\codea{\$\{iot:ClientId\}}$ of AWS\\
		$\RE$   & A set of resources $\rho$ \\
		$\Dev$ & A set of devices $d$ \\
		$\Lambda$ & Alphanumerical characters including '/'\\
		$\lang{\Lambda}$ & Language of topics \\
		$\Lambda \cup \{+,\#\}$ & Alphabet of client ids and topic filters \\
		$\lang{\Lambda \cup \{+,\#\}}$ & Language of client ids and topic filters\\
		$\effect \in \{\allow, \deny\}$ & Effect of a policy statement\\
		$\action \in \{\iotConn, \iotPub, \iotRec, \iotSub \}$ & IoT actions \\
		$s = \polState$ & Policy statement (\Cref{def:pol})\\
		$(\action, \rho)$ & Permission (\Cref{def:perm})\\
		$\certs$ & A set of certificates\\
		$\pols$ & A set of policies\\
		$\varphi\colon \certs \rightarrow \mathcal{P}(\pols)$ & A map from certificates to policies\\
		$k\colon \Dev \rightarrow \mathcal{P}(\certs)$ & A map from devices to certificates\\
		$(\certs, \pols, \varphi, \Dev, k)$ & IoT Configuration (\Cref{def:iot:config})\\
		$\conns$ & A set of active connections $\con$ \\
		$\inpols\colon \conns \rightarrow \pols$ & A map from connections to polices\\
		$\subs\colon \topicfs \rightarrow \mathcal{P}(\conns)$ & A map from topic filters to sub connections\\
		$\bstate = (\conns, \inpols, \subs)$ & Broker state (\Cref{def:broker:state})\\
		$a \in \{\mqttConn(c, \cid), \mqttSub(\tf), $ & \multirow{ 2}{*}{An MQTT request (\Cref{def:mqttreq})}\\[-.1cm]
		$\mqttDis, \mqttUns(\tf), \mqttPub(t) \}$ & \\
		$\msg = (\con, a)$ & MQTT message (\Cref{def:mqttreq})\\
		$\sigma \xrightarrow{\msg} \sigma'$ & Broker transition step (\Cref{fig:opSem})\\
		$\pi$ & Broker execution trace (\Cref{def:execs})\\
		$\dev \xRightarrow{t} \dev'$ & Communication triples\\
		$\varpi$ & Communication trace (\Cref{def:comtraces}) \\
		$\receivers(\sigma, t)$ & Set of receivers from $t$ in $\sigma$ (\Cref{def:recvr})\\
		$\Com(\pi)$ & Communication traces of $\pi$ (\Cref{def:extraces})\\
		$\Sin_c$, $\Sout_c$ & Set of I/O topics (\Cref{def:ioset})\\
		$G = (\certs \cup \topics, E)$ & Information-flow Graph (\Cref{def:info-graph})\\
		$\reqEval{(c, \cid, \action, \res)}$ & Permission predicate (\Cref{def:eval})\\
		$\Fout_{c}(t)$, $\Fin_{c}(t)$ & I/O Predicates (\Cref{def:iopred})\\
		\multirow{ 2}{*}{$G_s=(N, E)$} & Symbolic Information-flow graph\\[-.1cm]
		& (\Cref{def:info-graph:sym})\\
		$F_{c, c'}$ & Topic formula (\Cref{def:info-graph:sym})\\ 
		\bottomrule
	\end{tabular}
\end{table}

\section{Formal semantics}\label{sec:broker:sem}

Here, we formally define the evolution of an IoT system by introducing a formal semantics that captures the interaction between devices and the broker.
We derived our model from the AWS official documentation~\cite{AWS} and refined it by testing corner cases with handcrafted examples.

More precisely, given a configuration as defined in \Cref{sec:model}, we define a labelled transition system (LTS): 
its states represent the internal state of the broker, its labels are MQTT messages, and the transitions describe how the broker state changes upon reception of messages.
We start from the broker state, which
is characterised by the active connections, a mapping from connections to their corresponding policies, and a mapping from topic filters to the connections of the subscribed devices. 
\begin{definition}[Broker State]\label{def:broker:state}
	Given a configuration $(\certs, \pols, \varphi, \Dev, k)$, 
	a broker state is a triple $\bstate = (\conns, \inpols, \subs)$ where
	\begin{itemize}
	\item $\conns$ is a set of active connections;
	\item $\inpols\colon \conns \rightarrow \pols$ maps connections to their policies;
	\item $\subs\colon \topicfs \rightarrow \mathcal{P}(\conns)$ maps topic filters to 
	subscribed connections.
	\end{itemize}
\end{definition} 
Note that a single device may create multiple connections for authenticating with different certificates and client ids, so as to maximise its granted permissions.
In the following, we assume a function $\devof \colon \conns \rightarrow \Dev$ to associate each connection $\con$ with its corresponding device.
Moreover, when referring to $\inpols$ and $\subs$, we represent a function as a set of pairs $x \mapsto y$.
We write \emph{function update} as $\redef{f}{x}{y}$, meaning that the result returns $y$ when the input is $x$ and behaves like $f$ otherwise.

The labels of the LTS are \emph{MQTT messages} to the broker,
sent by devices over established connections (typically TCP) for publishing payloads or changing the broker state.
\begin{definition}[MQTT Message]\label{def:mqttreq}
Let $\con$ be a connection, an MQTT message over $\con$ is a pair $\msg = (\con, a)$ where $a$ is either
    \begin{itemize}
    \item $\mqttConn(c, \cid)$ - a $\code{CONNECT}$ request;
    \item $\mqttSub(\tf)$ - a $\code{SUBSCRIBE}$ request;
    \item $\mqttDis$ - a $\code{DISCONNECT}$ request;
    \item $\mqttUns(\tf)$ - an $\code{UNSUBSCRIBE}$ request;
    \item $\mqttPub(t)$ - a $\code{PUBLISH}$ request.
  \end{itemize}
\end{definition}

Given a configuration $(\certs, \pols, \varphi,  \Dev, k)$, a transition $\sigma \xrightarrow{\msg} \sigma'$ specifies how the broker state $\sigma$ evolves into $\sigma'$ upon the reception of the MQTT message $\msg$.
Transitions are defined by the rules of \Cref{fig:opSem} if some apply, or by the idle move $\bstate \xrightarrow{\msg} \bstate$ otherwise.
The rule \rulename{Con} states that when receiving a connection request, the broker 
authenticates the device association with the certificate;
instantiates the policy of the provided certificate with the declared client id;
asks the access control system if the obtained ground policy permits the connection;
finally, if the connection is permitted, it updates the set of connected devices and records the association of the connection with the policy.
The rule \rulename{Sub} states that when receiving a subscription request, the broker 
retrieves the policy associated with the connection;
asks the access control system if the obtained ground policy permits the subscription;
finally, if the permission is granted, it updates the set of subscribers to the topic filter by adding the connection.
The rule \rulename{Disc} simply states that when receiving a disconnection request, the broker removes all the occurrences of the connection from its state.
Finally, \rulename{UnSub} manages unsubscription requests by removing from the $\subs$ function the binding between the connection and the specific topic filter, if present.
All other messages do not cause a change in the broker states, either because it is not part of its intended semantics (e.g.~\mqttPub messages), or the request is denied by the access control mechanism, or the message is discarded (e.g.~if a subscription request arrives from a device that is not connected).

\begin{figure*}
\begin{gather*}
\footnotesize
\prftree[l]
	{\rulename{Con}}
		{\msg = (\con, \mqttConn(c, \cid))}
		{c \in k(\devof(\con))}
    	{P = \Phi(c)[\sfrac{\cid}{\cidVARa}]}
    	{P \permit (\iotConn, \cid)}
		{(\conns, \inpols, \subs)  \xrightarrow{\msg} (\conns \cup \{\con\}, \redef{\inpols}{\con}{P}, \subs)}
\qquad
\footnotesize
\prftree[l]
	{\rulename{Sub}}
		{\msg = (\con, \mqttSub(\tf))}
		{\con \in \conns}
		{\inpols(\con) \permit (\iotSub, \tf)}
		{(\conns, \inpols, \subs)  \xrightarrow{\msg} 
			(\conns, \inpols, \redef{\subs}{\tf}{\subs(\tf) \cup \{ \con \}})}
\\[.25cm]
\footnotesize
\prftree[l]
	{\rulename{UnSub}}
		{\msg = (\con, \mqttUns(\tf))}
		{\con \in \conns}
		{(\conns, \inpols, \subs)  
			\xrightarrow{\msg} 
			(\conns, 
			\inpols, 
			\redef{\subs}{\tf}{\subs(\tf) \setminus \{ \con \}})}
\qquad
\footnotesize
\prftree[l]
	{\rulename{Disc}}
		{\msg = (\con, \mqttDis)}
		{\conns' = \conns \setminus \{\con\}}
		{(\conns, \inpols, \subs) \xrightarrow{\msg} 
			(\conns', 
			\{ \con' \mapsto \inpols(\con') \mid \con' \in \conns' \}, 
			\{ \tf \mapsto \subs(\tf) \cap \conns' \mid \tf \in \topicfs \})}
\end{gather*}
\caption{Broker state evolution.}
\label{fig:opSem}
\end{figure*}

\begin{definition}[Broker Execution]\label{def:execs}
A broker execution $\pi$ is a finite sequence of transitions
\[
	\sigma_0 \xrightarrow{\msg_1} \sigma_1 \xrightarrow{\msg_2} \sigma_2 \xrightarrow{\msg_3} \dots \xrightarrow{\msg_n} \sigma_n
\]
We say that $\pi$ is \emph{feasible} if $\sigma_0$ is the initial \emph{empty} state $(\emptyset, \emptyset, \{ \tf \mapsto \emptyset \mid \tf \in \topicfs \})$.
\end{definition}

Hereafter, we denote with $\pi \odot \pi'$ the concatenation of executions, which is only defined when the last element of $\pi$ coincides with the first one of $\pi'$ (in the resulting execution, we omit the first broker state of $\pi'$ to avoid duplications).

\begin{example}\label{ex:exec}
Consider the light bulb $\code{light2}$ and presence sensor $\code{prsSens1}$, and let the broker state be $\sigma = (\{ \ell \}, \{ \ell \mapsto \pol \}, \subs_\emptyset)$: $\code{prsSens1}$ is connected with $\ell$, it can publish on any topic ($\pol = \{ (\allow, \iotPub, \{*\}) \}$), and there is no subscription ($\subs_\emptyset = \{ \tf \mapsto \emptyset \mid \tf \in \topicfs\}$).
Consider the following messages where: $\code{light2}$ connects with $\ell'$ (such that $\devof(\ell') = \code{light2}$); it subscribes to the topic filter $\code{phAC/floor1/dtdMovement/\#}$; and $\code{prsSens1}$ publishes on $\code{phAC/floor1/dtdMovement/light1}$.
\begin{align*}
\msg_1 = &(\ell', \mqttConn(c_{\codea{light2}}, \code{\#}))\\
\msg_2 = &(\ell', \mqttSub(\code{phAC/floor1/dtdMovement/\#}))\\
\msg_3 = &(\ell, \mqttPub(\code{phAC/floor1/dtdMovement/light1}))
\end{align*}
Note that $\code{light2}$ attempts a wildcard injection by using $\#$ as its client id.
The execution $\pi = \sigma \xrightarrow{\msg_1} \sigma_1 \xrightarrow{\msg_2} \sigma_2 \xrightarrow{\msg_3} \sigma_3$ has
\begin{align*}
\sigma_1 &= 
	(\{ \ell, \ell' \}, \{ \ell \mapsto \pol, \ell' \mapsto \pol' \}, \subs_\emptyset)\\
\sigma_2 &= 
	(\{ \ell, \ell' \}, \{ \ell \mapsto \pol, \ell' \mapsto \pol' \},\\
	&\ \ \{ \code{phAC/floor1/dtdMovement/\#} \mapsto \{\ell'\} \} \cup  \subs_\emptyset)
\end{align*}
where  $P'$ is the policy of \Cref{fig:light:bulb:policy} instantiated as:
\begin{align*}
     P' = \{&(\allow, \iotConn, \{*\}), (\allow, \iotRec, \{*\}),\\
    &(\allow, \iotSub, \{\code{phAC/floor?/dtdMovement/\#}\})\}
\end{align*}
During the first step, the condition $P' \permit (\iotConn, \code{\#})$ holds because 
$\code{\#} \in \lang{\ast^{\AWS}}$: the chosen id is included in the resource expression of the first allow statement
 (and there is no deny statement).
In the second step, subscription is permitted by $P'$ because $\code{floor1} \in \lang{\code{floor?}^{\AWS}}$.
Finally, publication causes an idle move,
leaving the broker state unchanged 
($\sigma_3 = \sigma_2$).
\end{example}

Since we are interested in 
how information propagates between devices through the broker, we introduce communication triples $\dev \xRightarrow{t} \dev'$, meaning that $\dev$ has communicated some information to $\dev'$ through the topic $t$.  
We compose communication triples into communication traces.

\begin{definition}[Communication Trace]\label{def:comtraces}
A communication trace $\varpi$ from $\dev_0$ to $\dev_n$ is a finite sequence of communication triples
    \[
	\dev_0 \xRightarrow{t_1} \dev_1 \xRightarrow{t_2} \dev_2 \xRightarrow{t_3} \dots \xRightarrow{t_n} \dev_n.
	\]
\end{definition}
We extend $\odot$ to communication traces as expected, and
we abuse notation by writing $\Pi \odot \Pi'$ for the set containing all the elements of $\Pi$ and $\Pi'$, as well as all the defined concatenations $\varpi \odot \varpi'$ with some $\varpi \in \Pi$ and $\varpi' \in \Pi'$.
Intuitively, $\Pi \odot \Pi'$ encodes all the possible communication by considering the ones that happen separately on the two sets, as well as their concatenation when the final receiver of a trace in $\Pi$ is the first sender of another trace in $\Pi'$.

We now show how to derive communication traces from broker executions.
First, we introduce the auxiliary notion of the set of devices that can receive messages published over a topic $t$ when the broker is in a state $\sigma$.
Formally,
\begin{definition}\label{def:recvr}
Given a broker state $\sigma = (\conns, \inpols, \subs)$ and a topic $t$, the set $\receivers(\sigma, t) \subseteq \Dev$ is defined as 
\[
\receivers(\sigma, t) =\!\! \bigcup_{\tf \in \topicfs\downarrow_t}\!\! \{ \devof(\con) \mid \con \in \subs(\tf) \text{ and } \inpols(\con) \permit (\iotRec, t) \}
\]
where $\topicfs\!\downarrow_t = \{\tf \mid t \in \lang{\mathit{tf}^{\mqtt}}\}$.
\end{definition}
We can now define the communication traces resulting from a given execution $\pi$:
we consider each publish request in $\pi$ and compute its receivers by taking into account the broker state as of when it was received.
\begin{definition}[Traces of an Execution]\label{def:extraces}
	The set $\Com(\pi)$ of the communication traces of an execution $\pi$ is defined as
    \begin{align*}
		\Com(\pi) = 
        \begin{cases}
			\{ \dev \xRightarrow{t} \dev' \mid \dev = \devof(\con) &\text{if } \pi = \sigma \xrightarrow{(\con, \mqttPub(t))} \sigma' \\[-1.25mm]
            \quad\text{ and }\dev' \in \receivers(\sigma, t) \} &\quad\text{ with }\sigma = (\conns, \inpols, \subs)\\
            &\quad\text{ and }\inpols(\con) \permit (\iotPub, t)\\[1mm]
        \Com(\sigma \xrightarrow{\msg} \sigma') \odot \Com(\pi') &\text{if } \pi = \sigma \xrightarrow{\msg} \sigma' \odot \pi'\\[1.25mm]
		\emptyset	& \text{otherwise }
        \end{cases}
   	\end{align*}
\end{definition} 

\begin{example}\label{ex:request}
Let $t$ be $\code{phAC/floor1/dtdMovement/light1}$.
The trace of the broker execution $\pi$ of~\Cref{ex:exec} is
\begin{gather*}
\Com(\pi) = \emptyset \odot \emptyset \odot \Com(\sigma_2 \xrightarrow{\msg_3} \sigma_3) = \{ \code{prsSens1} \xRightarrow{t} \code{light2} \}
\end{gather*}
Note that the last message is a publish request, that the sender is $\code{prsSens1}$, and that $\code{light2}$ is subscribed to $\tf = \code{phAC/floor1/dtdMovement/\#}$. Thus, $\code{light2}$ is among the receivers because $t \in \lang{\tf^{\mqtt}}$.
\end{example}

The formal development above enables us to prove the correctness of our verification framework of \Cref{sec:infoflow}.
More precisely, the following theorem ensures that the information flow graph is 
a correct and complete representation of the information flow caused by the traces of broker executions,
also for reachability through multi-step communication traces.
Thus, we can use the graph to analyse the security properties of confidentiality, integrity, and isolation in our attacker model.
\begin{restatable}{theorem}{optoflowf}
\label{th:optoflowf}
    Given a configuration and its information flow graph $G$, for every devices $d$ and $d'$ it holds that
	$G$ has an information flow from $d$ to $d'$
	if and only if a feasible execution $\pi$ exists such that $\Com(\pi)$ contains a trace from $\dev$ to $\dev'$.
\end{restatable}

\section{Proofs}\label{ap:cost}

\newcommand{\sbot}{\sigma_{\bot}}

We first prove that our symbolic characterization of the information flow graph is correct.
Then, we ensure consistency of the permitted information flows described by the operational semantics with the paths of the information flow graph.
Finally, we prove the complexity of~\Cref{al:sifg}.

\subsection{Consistency of the Symbolic Information Flow Graph}

We establish the following consistency result between the intentional description of the policy interrogation of~\Cref{def:perm} and the predicate of~\Cref{def:eval}.
\begin{lemma}\label{lm:setpred}
	Given a configuration $I$, the permission predicate $\reqEval{\req}$ is true if and only if $\Phi(c)[\sfrac{\cid}{\cidVAR}] \permit (\alpha, \rho)$.
\end{lemma}
\begin{proof}
	The result trivially holds by definition, and can be proved by induction on the cardinality of $\Phi(c)$.
\end{proof}

The following lemma ensures that two devices are reachable in one step in the information flow graph if and only if the same is true in the symbolic version.
\begin{lemma}\label{lem:arcs-exist}
	Given a configuration, let $G = (N,E)$ be its information flow graph, and $G_s=(N_s, E_s)$ the symbolic version.
	Let $d, d'$ be a pair of devices, then the two following statements are equivalent:
	\begin{enumerate}
		\item a topic $t$ exists such that $\{(d,t), (t, d')\} \subseteq E$;
		\item two certificates $c$, $c'$ exist such that $\{(d,F_{c,c'}), (F_{c,c'}, d')\} \subseteq E_s$.
	\end{enumerate}
\end{lemma}
\begin{proof}
	We first note that, by \Cref{def:info-graph}, 
	the first statement holds if and only if 
	there exists a $t$ such that $t \in \Sout_{c} \cap \Sin_{c'}$.
	Notice that~\Cref{lm:setpred} implies:
	\[
	t \in \Sout_c \text{ iff }\Fout_c(t) 
	\qquad \text{and} \qquad 
	t \in \Sin_{c'} \text{ iff }\Fin_{c'}(t).
	\]
	By~\Cref{def:info-graph:sym}, 
	this coincides with $t$
	satisfying $F_{c,c'}$, and by quantifier commutativity it is the same as the second statement.
\end{proof}

The following generalizes the previous result to multi-step reachability.
\correct*
\begin{proof}
The statement can be proved by induction on the length of the information flow, resorting to~\Cref{lem:arcs-exist} for proving that single steps coincides.
\end{proof}

\subsection{Consistency of Semantics and Information Flow Graph}

\begin{notation}
In the following, we write:
\begin{itemize}
    \item $\varsigma_\bot$ for the function $\{ \tf \mapsto \emptyset \mid \tf \in \topicfs \}$;
    \item $\sbot$ for the initial broker state $(\emptyset, \emptyset, \varsigma_\bot)$;
    \item $\sigma \rightarrow^{*} \sigma'$ when $\sigma \xrightarrow{\msg_1} \sigma_1 \xrightarrow{\msg_2} \dots \xrightarrow{\msg_n} \sigma'$ for some $\msg_1, \dots \msg_n$ and $\sigma_1, \dots \sigma_{n-1}$;
    \item $\dev \rightarrow^{*} \dev'$ when $\dev \xrightarrow{t_1} \dev_1 \xrightarrow{t_2} \dots \xrightarrow{t_n} \dev'$ for some $t_1, \dots t_n$ and $\dev_1, \dots \dev_{n-1}$;
\end{itemize}
\end{notation}

The next auxiliary results proves that necessary conditions applies to reachability of some given broker states.
\begin{lemma}\label{lm:con}
    Let $\sigma$ be a broker state such that $\sbot \rightarrow^{*} \sigma = (\conns, \inpols, \subs)$, and let $\con \in \conns$, then $\inpols(\con) = \Phi(c)[\sfrac{\cid}{\cidVAR}]$ for some $c$ and $\cid$ such that $c \in k(\devof(\con))$ and $\Phi(c)[\sfrac{\cid}{\cidVAR}] \permit (\iotConn, \cid)$.
\end{lemma}
\begin{proof}
We prove the property by induction on the length of the execution $\pi$ from $\sbot$ to $\sigma$.
The state $\sbot$ vacuously satisfies the lemma, as no $\con$ is in $\conns = \emptyset$.
Assume  $\pi = \pi' \odot (\sigma \xrightarrow{\msg} \sigma')$ where $\sigma = (\conns, \inpols, \subs)$ and $\sigma' = (\conns', \inpols', \subs')$.
Then, by induction hypothesis, if $\con \in \conns$, then $\inpols(\con) = \Phi(c)[\sfrac{\cid}{\cidVAR}]$ for some $c$ and $\cid$ such that $c \in k(\devof(\con))$ and $\Phi(c)[\sfrac{\cid}{\cidVAR}] \permit (\iotConn, \cid)$.
We assume that 
$\con' \in \conns'$, and
show by cases on the rules in~\Cref{fig:opSem} that the same properties hold for $\sigma'$, namely
$\inpols'(\con') = \Phi(c')[\sfrac{\cid'}{\cidVAR}]$ for some $c'$ and $\cid'$ such that $c' \in k(\devof(\con'))$ and $\Phi(c')[\sfrac{\cid'}{\cidVAR}] \permit (\iotConn, \cid')$.
\begin{description}
\item[Case \rulename{Con}:] for $\con' \in \conns'$ to be true it must be that either $\con' \in \conns$, or
	$\con'$ is the connection in the message $\msg$.
	In the former case, the thesis follows by induction hypothesis since $\redef{\inpols}{\con'}{P}(\con) = \inpols(\con)$. 
	In the latter case, the thesis holds thank to the premises of the derivation rule, which coincide with our desiderata.
\item[Cases \rulename{Sub} and \rulename{unSub}:] notice that $\conns' = \conns$ and $\inpols' = \inpols$, thus for $\con' \in \conns'$ to be true it must be that $\con' \in \conns$ and we can apply the induction hypothesis.
\item[Case \rulename{Disc}:] induction hypothesis suffices because it must be that $\con' \in \conns$, and it holds by construction that $\inpols'(\con'') =  \inpols(\con'')$ for each $\con''$ for which both functions are defined.
\item[Default idel move:] if not rule applies then $\sigma = \sigma'$ and the result trivially follows from induction hypothesis.\qedhere
\end{description}
\end{proof}

\begin{lemma}\label{lm:pol}
    Let $\sigma$ be a broker state such that $\sbot \rightarrow^{*} \sigma = (\conns, \inpols, \subs)$, if $\inpols(\con)$ is defined, then $\con \in \conns$. 
\end{lemma}
\begin{proof}
We prove the property by induction on the length of the execution $\pi$ from $\sbot$ to $\sigma$.
The state $\sbot$ vacuously satisfies the lemma, as the domain of $\inpols$ is empty.
Assume  $\pi = \pi' \odot (\sigma \xrightarrow{\msg} \sigma')$ where $\sigma = (\conns, \inpols, \subs)$ and $\sigma' = (\conns', \inpols', \subs')$, and consider the rules in~\Cref{fig:opSem}.
For the \rulename{Con} rule, $\inpols'(\con')$ is defined only if either $\inpols(\con')$ is defined, or if $\con'$ is the connection in the $\msg$ message.
In the former case, the thesis follows by induction hypothesis; in the latter, it holds by construction.
If the rule \rulename{Disc} is used to build the transition, then $\inpols'(\con')$ defined implies that also $\inpols(\con')'$ is also defined, and $\con' \neq \con''$ with $\con''$ the connection in $\msg$.
Then the thesis follows by induction hypothesis.
The other cases hold by induction hypothesis, because $\inpols' = \inpols$ and $\conns' = \conns$.
\end{proof}

\begin{lemma}\label{lm:sub}
    Let $\sigma$ be a broker state such that $\sbot \rightarrow^{*} \sigma = (\conns, \inpols, \subs)$, and let $\con \in \subs(\tf)$, then 
    $\inpols(\con) = \Phi(c)[\sfrac{\cid}{\cidVAR}]$ for some $c$ and $\cid$ such that 
    $\Phi(c)[\sfrac{\cid}{\cidVAR}] \permit (\iotSub, \tf)$, $c \in k(\devof(\con))$ and $\Phi(c)[\sfrac{\cid}{\cidVAR}] \permit (\iotConn, \cid)$.
\end{lemma}
\begin{proof}
We prove the property by induction on the length of the execution $\pi$ from $\sbot$ to $\sigma$.
The state $\sbot$ vacuously satisfies the lemma.
Assume  $\pi = \pi' \odot (\sigma \xrightarrow{\msg} \sigma')$ where $\sigma = (\conns, \inpols, \subs)$ and $\sigma' = (\conns', \inpols', \subs')$.
Then, by induction hypothesis, if $\con \in \subs(\tf)$, then 
$\inpols(\con) = \Phi(c)[\sfrac{\cid}{\cidVAR}]$ for some $c$ and $\cid$ such that 
$\Phi(c)[\sfrac{\cid}{\cidVAR}] \permit (\iotSub, \tf)$, $c \in k(\devof(\con))$ and $\Phi(c)[\sfrac{\cid}{\cidVAR}] \permit (\iotConn, \cid)$.

We assume that 
$\con' \in \subs'(\tf')$, and
show by cases on the rules in~\Cref{fig:opSem} that the same properties hold for $\sigma'$, namely
$\inpols'(\con') = \Phi(c')[\sfrac{\cid'}{\cidVAR}]$ for some $c'$ and $\cid'$ such that 
$\Phi(c')[\sfrac{\cid'}{\cidVAR}] \permit (\iotSub, \tf')$, $c' \in k(\devof(\con'))$ and $\Phi(c')[\sfrac{\cid'}{\cidVAR}] \permit (\iotConn, \cid')$.
\begin{description}
\item[Case \rulename{Con}:] the result follows from the fact that $\subs = \subs'$, thus $\con' \in \subs'(\tf')$ implies $\con' \in \subs(\tf)$.
In addition, $\con''$ of $\msg$ is fresh in $\conns$, thus $\inpols(\con') = \redef{\inpols'}{\ell'}{P}(\con') = \inpols(\con')$.
\item[Case \rulename{Sub}:] 
if the topicfilter $\tf''$ in $\msg$ is $\tf'$, then from~\Cref{lm:con} it holds that $\con' \in \conns$ implies 
$\inpols(\con') = \inpols'(\con') = \Phi(c')[\sfrac{\cid'}{\cidVAR}]$ for some $c'$ and $\cid'$ such that $c' \in k(\devof(\con'))$ and $\Phi(c')[\sfrac{\cid'}{\cidVAR}] \permit (\iotConn, \cid')$.
Finally, the last condition of the derivation rule guarantees that $\inpols'(\con') \permit (\iotSub, \tf')$.

If $\tf''$ in $\msg$ is not $\tf'$ instead, it must be that $\tf'$ is some $\tf$ such that $\con' \in \subs(\tf)$ an the result follows by induction hypothesis.

\item[Cases \rulename{unSub}, \rulename{Disc}, and default idel move] 
Induction hypothesis suffices because it must be that $\con' \in \subs(\tf')$, and it holds by construction that $\inpols'(\con'') =  \inpols(\con'')$ for each $\con''$ for which both functions are defined.\qedhere
\end{description}
\end{proof}

We now prove correctness and completeness of the information flow graph separately.
\begin{lemma}\label{lm:optoflowif}
Given a configuration $I$, two devices $d, d' \in \Dev$ and a topic $t$,
if a feasible execution $\pi$ exists such that $d \xRightarrow{t} d' \in \Com(\pi)$, then the information flow graph of $I$ contains both
$(d, t)$ and $(t, d')$.
\end{lemma}
\begin{proof}
    Take $\dev, \dev', t$ and $\pi$, such that $d \xRightarrow{t} d' \in \Com(\pi)$.
    We prove the desideratum by induction over the conditions of~\Cref{def:extraces}.
    
    As base case, let $\pi = \sigma \xrightarrow{\msg} \sigma'$, with $\sigma = (\conns, \inpols, \subs)$, $\msg = (\con, \mqttPub(t))$, and 
    $\inpols(\con) \permit (\iotPub, t)$.
    Since $d \xRightarrow{t} d' \in \Com(\pi)$, it must be that $d = \devof(\con)$ and $\dev' \in \receivers(\sigma, t)$.
    By~\Cref{lm:pol}, $\con \in \conns$, and thus       
	it holds by~\Cref{lm:con} that 
	$\inpols(\con) = \Phi(c)[\sfrac{\cid}{\cidVAR}]$ for some $c$ and $\cid$ such that $c \in k(d)$ and $\Phi(c)[\sfrac{\cid}{\cidVAR}] \permit (\iotConn, \cid)$.
	Therefore, we know that $t \in \Sout_{c}$ by~\Cref{def:ioset}, and thus, by~\Cref{def:info-graph}, $(d, t)$ is an arc of the information flow graph.
		
	Moreover, by~\Cref{def:recvr}, $\dev' \in \receivers(\sigma, t)$ implies $t$ is in $\lang{\tf^{\mqtt}}$ for some $\tf$ such that $\con' \in \subs(\tf)$ and $\inpols(\con) \permit (\iotRec, t)$ for some $\con'$ such that $\devof(\con') = \dev'$.
	Since $\con' \in \subs(\tf)$, by~\Cref{lm:sub}, a certificate $c'$ must exists such that $\inpols(\con) = \Phi(c')[\sfrac{\cid'}{\cidVAR}]$ and $\Phi(c')[\sfrac{\cid'}{\cidVAR}] \permit (\iotSub, \tf)$ for some $\cid'$ such that $c' \in k(d')$ and $\Phi(c')[\sfrac{\cid'}{\cidVAR}] \permit (\iotConn, \cid')$.
	Therefore, we know that $t \in \Sin_{c'}$ by~\Cref{def:ioset}, and thus, by~\Cref{def:info-graph}, $(t, d')$ is an arc of the information flow graph.
	
	Finally, we consider the induction step and assume $\pi = \sigma \xrightarrow{\msg} \sigma' \odot \pi'$.
	Then, since $d \xRightarrow{t} d'$ is a single step, it is either that $d \xRightarrow{t} d' \in \Com(\sigma \xrightarrow{\msg} \sigma')$ or $d \xRightarrow{t} d' \in \Com(\pi')$, and in both cases we can directly rely on the induction hypothesis.
\end{proof}

\begin{lemma}\label{lm:optoflowfi}
Given a configuration and its information flow graph $G$, 
two devices $d, d' \in \Dev$, and a topic $t$,
if both $(d, t)$ and $(t, d')$ are arcs in $G$, then
a feasible  execution $\pi$ exists such that $d \xRightarrow{t} d' \in \Com(\pi)$.
\end{lemma}
\begin{proof}
Take $d, d'$ and $t$, and assume both $(d, t)$ and $(t, d')$ are arcs of the information flow graph.
By~\Cref{def:info-graph}, 
$c, c'$ exists such that $c \in k(d)$, $c' \in k(d')$, with  
$t \in \Sout_c$ and $t \in \Sin_{c'}$; thus we know that $\cid$, $\cid'$ and $\tf$ exists such that
we can build the following feasible execution $\pi$ in accordance with the rules of~\Cref{fig:opSem}, where $\devof(\con) = \dev$, $\devof(\con') = \dev'$, and $\Phi_{c, \cid}$ stands for $\Phi(c)[\sfrac{\cid}{\cidVARa}]$.
\begin{align*}
\sbot &\xrightarrow{(\con, \mqttConn(c, \cid))}\\
&(\{ \con \}, \{\con \mapsto \Phi_{c, \cid} \}, \varsigma_\bot)\\
&\xrightarrow{(\con', \mqttConn(c', \cid'))} \\
&(\{ \con, \con' \}, \{\con \mapsto \Phi_{c, \cid}, \con' \mapsto \Phi_{c', \cid'}\}, \varsigma_\bot)\\
&\xrightarrow{(\con', \mqttSub(\tf))} \\
&(\{ \con, \con' \}, \{\con \mapsto \Phi_{c, \cid}, \con' \mapsto \Phi_{c', \cid'}\}, \{ \tf \mapsto \{ \con' \} \})\\
&\xrightarrow{(\con, \mqttPub(t))} \\
&(\{ \con, \con' \}, \{\con \mapsto \Phi_{c, \cid}, \con' \mapsto \Phi_{c', \cid'}\}, \{ \tf \mapsto \{ \con' \} \})
\end{align*}

The proof concludes by showing that $d \xRightarrow{t} d' \in \Com(\pi)$.
In particular, it holds by~\Cref{def:extraces} that the conditions of $t \in \Sout_c$ and $t \in \Sin_{c'}$ imply $d \xRightarrow{t} d' \in \Com(\sigma \xrightarrow{(\con, \mqttPub(t))} \sigma)$, with $\sigma$ the last broker state of $\pi$.
\end{proof}

\begin{lemma}
\label{th:optoflow}
	Given a configuration, its information flow graph $G$ and two devices $d, d'$, the graph $G$ contains both the arcs $(d, t)$ and $(t,d')$ if and only if a feasible execution $\pi$ exists such that  $\dev \xRightarrow{t} \dev' \in \Com(\pi)$.
\end{lemma}
\begin{proof}
    By~\Cref{lm:optoflowif} and~\ref{lm:optoflowfi}.
\end{proof}

We address now reachability over the information flow graph, and its correspondence to multi-steps communication traces.
A useful property of feasible executions is that they can be composed, thanks to $\mqttDis$ which allows extending each feasible execution to one from $\sbot$ to $\sbot$.
\begin{lemma}\label{lm:compose}
Let $\pi, \pi'$ be feasible executions of a configuration, then a feasible execution $\pi''$ exists such that $\Com(\pi'') = \Com(\pi) \odot \Com(\pi')$.
\end{lemma}
\begin{proof}
Let $\pi$ and $\pi'$ be as follows:
\begin{align*}
    \pi &= \sbot \xrightarrow{\msg_1} \sigma_1 \xrightarrow{\msg_2} \dots \xrightarrow{\msg_n} \sigma_n\\
    \pi' &= \sbot \xrightarrow{\msg_1'} \sigma_1' \xrightarrow{\msg_2'} \dots \xrightarrow{\msg_m'} \sigma_m'
\end{align*}
Let also $\sigma_n = (\conns, \inpols, \subs)$, with $\conns = \{ \con_i \}_{i = 1}^k$ and let $\msg''_i = (\con_i, \mqttDis)$ for all $i$ from $1$ to $k$.
Then it holds by construction that the following $\pi''$ is a feasible execution:
\begin{align*}
    \pi'' &= \sbot \xrightarrow{\msg_1} \sigma_1 \xrightarrow{\msg_2} \dots \xrightarrow{\msg_n} \\
    &\sigma_n \xrightarrow{\msg_1''} \sigma_1'' \xrightarrow{\msg_2''} \sigma_2'' \dots \xrightarrow{\msg_k''} \\
    &\sbot \xrightarrow{\msg_1'} \sigma_1' \xrightarrow{\msg_2'} \dots \xrightarrow{\msg_m'} \sigma_m'.
\end{align*}
Moreover, by~\Cref{def:extraces}, $\Com(\sigma_n \xrightarrow{\msg_1''} \dots \xrightarrow{\msg_k''}) = \emptyset$, and also $\Com(\pi'') = \Com(\pi) \odot \emptyset \odot \Com(\pi') = \Com(\pi) \odot \Com(\pi')$. 
\end{proof}

Moreover, if a trace of a feasible execution can be decomposed (i.e., represented as the composition of two traces), then feasible executions exists also for the components.
\begin{lemma}\label{lm:decompose}
	Let $\pi$ and $\varpi$ be a feasible execution of a configuration and a trace such that $\varpi' \odot \varpi'' = \varpi \in \Com(\pi)$ for some $\varpi'$ and $\varpi''$, then feasible executions $\pi', \pi''$ exist such that $\varpi' \in \Com(\pi')$ and $\varpi'' \in  \Com(\pi'')$.
\end{lemma}
\begin{proof}
Let $\pi\setminus \mqttPub$ be the execution $\pi$ where each occurrence of $\mqttPub$ messages are removed, and notice the following trivial properties: $\pi\setminus \mqttPub$ is feasible if $\pi$ is feasible; $\Com(\pi\setminus \mqttPub) = \emptyset$; the final states of $\pi$ and $\pi\setminus \mqttPub$ coincide.

By~\Cref{def:extraces}, $\varpi' \odot \varpi''$ implies that $\pi_0'$ and $\pi_0''$ exists such that
$\pi = \pi_0' \odot \pi_0''$, with $\varpi' \in \Com(\pi_0')$ and $\varpi'' \in  \Com(\pi_0'')$.
Note that $\pi_0'$ is feasible, hence we can take $\pi' = \pi_0'$.
For $\pi''$, we take instead $\pi'' = (\pi_0'\setminus \mqttPub) \odot \pi_0''$.
The result follows by noticing that $\pi''$ is well defined because the final state of $(\pi_0'\setminus \mqttPub)$ coincides with the one of $\pi_0'$, that $\pi''$ is feasible because $\pi_0'$ is feasible, and that $\Com((\pi_0'\setminus \mqttPub) \odot \pi_0'') = \emptyset \odot \Com(\pi_0'') = \Com(\pi_0'')$.
\end{proof}

\begin{table*}[th]
    \caption{\thetool\ Scalability on real-world configurations: comparison between cvc5 and Z3 as backend SMT solver. }\label{tab:rw:z3}
    \centering
    \footnotesize
    \setlength{\tabcolsep}{7pt}
    \begin{tabular}{c l | c c c c | c c | c c c}
		\toprule
		& &  \textbf{nodes} & \textbf{hard strategies}  & \multicolumn{6}{c}{\textbf{building time with ($s$)}} & \textbf{query time ($s$)} \\
        & & \textbf{avg.}   &  \textbf{avg.}   & \multicolumn{2}{c}{\textbf{min.}} & \multicolumn{2}{c}{\textbf{avg.}}  & \multicolumn{2}{c}{\textbf{max.}} & \textbf{avg.} \\
        & & &  & \textbf{cvc5} & \textbf{Z3} & \textbf{cvc5} & \textbf{Z3} & \textbf{cvc5} & \textbf{Z3} & \\
		\midrule
   		\multirow{ 13}{*}{\rotatebox{90}{\footnotesize \textbf{certificates}}} 
        &    20 	&   264.6    &  1.8  &  0.42   &   1.22  &  2.87  &   4.94  & 12.00  &  10.69  & 0.0060 \\
        &    40 	&   973.0	 &  4.5  &  2.30   &   6.17  &  6.28  &  12.64  & 16.55  &  20.28  & 0.0086 \\
        &    60 	&  2211.2	 &  5.6	 &  2.93   &   7.04  &  9.22  &  18.22  & 23.39  &  30.00  & 0.0108 \\
        &    80 	&  3903.9	 & 13.6	 &  5.93   &  18.67  & 16.34  &  30.72  & 29.08  &  47.88  & 0.0136	\\
        &    100 &  6098.1	 & 16.7	 &  7.50   &  18.89  & 19.46  &  39.82  & 36.55  &  58.76  & 0.0167	\\
        &    120 &  8764.2	 & 21.5	 & 13.58   &  35.17  & 24.69  &  52.35  & 41.49  &  87.98  & 0.0198	\\
        &    140 & 12080.1	 & 28.7	 & 18.61   &  42.75  & 31.06  &  64.81  & 46.63  & 100.06  & 0.0244	\\
        &    160 & 15415.9	 & 29.6	 & 20.03   &  51.89  & 33.01  &  72.11  & 44.43  &  90.16  & 0.0280	\\
        &    180 & 19571.7	 & 37.0	 & 26.61   &  69.91  & 37.87  &  86.89  & 52.79  & 109.81  & 0.0318	\\
        &    200 & 24209.9    & 45.8	 & 31.62   &  78.76  & 46.06  & 103.63  & 57.34  & 125.65  & 0.0349	\\
        &    220 & 29066.2    & 50.5	 & 32.39   &  95.42  & 50.93  & 117.21  & 62.02  & 129.81  & 0.0393	\\
        &    240 & 34732.6    & 59.0	 & 42.11   & 114.48  & 57.68  & 130.17  & 62.41  & 141.43  & 0.0430	\\
        &    258 & 40019.0    & 65.0	 & 62.35   & 141.24  & 62.67  & 144.20  & 63.53  & 152.19  & 0.0473	\\

		\bottomrule
    \end{tabular}
\end{table*}

\optoflowf*
\begin{proof}
Assume $G$ has an information flow from $d$ to $d'$, i.e. there is are topics $t_1, t_2, \dots, t_n$ and devices $d_1, d_2, \dots, d_n, d_{n+1}$ such that $(d_i, t_i)$ and $(t_i, d_{i+1})$ are arcs of $G$, with $d_1 = d$ and $d_{n+1} = d'$.
Then, by~\Cref{th:optoflow}, for each $i = 1, \dots, n$, a feasible executions $\pi_i$ exist such that $d_i \xRightarrow{t_i} d_{i+1} \in \Com(\pi_i)$.
Finally, by applying~\Cref{lm:compose} $n$ times, a feasible $\pi$ exists such that $\Com(\pi) = \Com(\pi_1) \odot \Com(\pi_2) \dots \odot \Com(\pi_i)$ contains the trace 
$d_1 \xRightarrow{t_1} d_2 \xRightarrow{t_2} \dots d_n \xRightarrow{t_n} d_{n+1}$, which is from $d$ to $d'$ by construction.

Assume instead that a feasible execution $\pi$ exists such that $\Com(\pi)$ contains a trace from $\dev$ to $\dev'$.
We prove that $G$ has an information flow from $d$ to $d'$ by induction on the length of the trace $\varpi$ from $d$ to $d'$.
If $\varpi$ is a single step, then the information flow exists in $G$ by~\Cref{th:optoflow}.
Otherwise, assume that $\varpi = \varpi' \odot \varpi''$, then by~\Cref{lm:decompose} two feasible executions $\pi'$ and $\pi''$ exist with  $\varpi' \in \Com(\pi')$ and $\varpi'' \in  \Com(\pi'')$.
By induction hypothesis, $G$ contains both information flows for $\varpi'$ (starting with $d$) and $\varpi''$ (ending in $d'$), and by construction the two can compose (the final device of $\varpi'$ must be the initial one of $\varpi''$).
\end{proof}

\subsection{Complexity of our Approach}

We prove the worst case complexity of our algorithm for building the symbolic information flow graph.
In spite of the theoretical result, our heuristics make the approach efficient when dealing with real-world configurations, as shown by our experimental evaluation.

\complexity*
\begin{proof}
\Cref{al:sifg} iterates over each pair of certificates $(c, c')$ and for each pairs it builds the predicate $F_{c,c'}$ and checks its satisfiability.
The overall cost is thus $O(|\certs|^2\cdot (\mathit{B} + \mathit{S} + \mathit{U}))$, where $\mathit{B}$, $\mathit{S}$ and $\mathit{U}$ are the cost of building $F_{c,c'}$, of checking its satisfiability and of updating the graph,
respectively.

To estimate the value of $B$, consider a request $r$ with certificate $c$. 
From~\Cref{def:eval} we have that the formula $\reqEval{r}$ is of length at most $O(\sum_{s = (\effect, \action, \RE) \in \Phi(c)} |\RE|)$.
We can approximate the expression above as $O(\overline{\pol}\cdot \overline{\RE})$ by considering the worst case, i.e. the policy containing the greatest number of statements $\overline{\pol}$ and the statement with the greatest number of resource expressions $\overline{\RE}$.
Although each $F_{c,c'}$ is the combination of five expressions built though $\reqEval{\_}$, its length is still at most $O(\overline{\pol}\cdot \overline{\RE})$ (see~\Cref{def:iopred} and ~\Cref{def:info-graph:sym}).
The cost of building $F_{c,c'}$ is thus $\mathit{B} = O(\overline{\pol}\cdot \overline{\RE})$.

To estimate the value of $S$, consider that checking the satisfiability of $F_{c,c'}$ through \Cref{al:witness} depends on the decision procedure used by the solver, which is invoked a fixed number of times (ten times in the worst case).
Overall, the worst case for $S$ is $O(f(\overline{\pol}\cdot \overline{\RE}))$, where $f$ is the cost of the algorithm used by the SMT solver for dealing with regular expressions and string concatenation.

To estimate the value of $U$, consider the worst case: every certificate is associated with all the devices.
Then we iterate on every device twice, one for each internal loop of~\Cref{al:sifg}.
Therefore, $\mathit{U}$ is bounded by $O(|\Dev|)$.
\end{proof}

\section{Comparing different SMT solvers}\label{app:z3}
\Cref{tab:rw:z3} compares the scalability of \thetool\ using cvc5 and Z3 as the SMT solver on the experiments described in \Cref{sec:tool}.
The experiments of \Cref{tab:rw:z3} were carried out by generating a set of random devices from which we derive the configuration on which we tested the two versions of the tool. 
The solutions found by using the two solvers coincide as expected, so do the average number of nodes and the times the hard strategy is required. 
Our results show that cvc5 performs better than Z3 almost all the time.

\end{document}